\documentclass[a4paper]{article}

\usepackage{amsmath, amssymb}
\usepackage{graphicx, subfigure}
\usepackage{booktabs}
\usepackage{rotating}
\usepackage[authoryear,round]{natbib}

\newcommand{\keywords}[1]{{\bf Keywords:} {#1}}

\newcommand{\JEL}[1]{{\bf JEL:} {#1}}

\newtheorem{thm}{theorem}[section]

\newtheorem{proposition}[thm]{Proposition}

\newtheorem{lemma}[thm]{Lemma}
\newenvironment{proof}{{\bf Proof. }}{\hfill$\Box$}
\newenvironment{MyProof}[1]{{\bf Proof {#1}}}{\hfill$\Box$}

\newcommand{\MakeTitle}{\maketitle\newcommand{\and}{$\cdot$ }}

\title{From the Samuelson Volatility Effect to a Samuelson Correlation Effect: Evidence from Crude Oil Calendar Spread Options
	\thanks{
		We would like to thank
        Iain Clark, Jean-Baptiste Gheeraert, Cassio Neri, Damien Pons and Matthias Scherer
        for helpful and stimulating comments, discussions and suggestions.
	}
}

\author{
	Lorenz Schneider
	\thanks{Center for Financial Risks Analysis (CEFRA), EMLYON Business School, \texttt{schneider@em-lyon.com}.}
    \quad \quad
	Bertrand Tavin
	\thanks{Center for Financial Risks Analysis (CEFRA), EMLYON Business School, \texttt{tavin@em-lyon.com}.}
}

\date{\today}

\begin{document}

\MakeTitle


\begin{abstract}
We introduce a multi-factor stochastic volatility model based on the CIR/Heston stochastic volatility process.
In order to capture the Samuelson effect displayed by commodity futures contracts, we add expiry-dependent exponential damping factors to their volatility coefficients.
The pricing of single underlying European options on futures contracts is straightforward and can incorporate the volatility smile or skew observed in the market.
We calculate the joint characteristic function of two futures contracts in the model in analytic form
and use the one-dimensional Fourier inversion method of \citet{CaldanaFusai2013} to price calendar spread options.
The model leads to stochastic correlation between the returns of two futures contracts.
We illustrate the distribution of this correlation in an example.
We then propose analytical expressions to obtain the copula and copula density directly from the joint characteristic function of a pair of futures.
These expressions are convenient to analyze the term-structure of dependence between the two futures produced by the model.
In an empirical application we calibrate the proposed model to volatility surfaces of vanilla options on WTI.
In this application we provide evidence that the model is able to produce the desired stylized facts in terms of volatility and dependence.
In a separate appendix, we give guidance for the implementation of the proposed model and the Fourier inversion results by means of one and two-dimensional FFT methods.


\bigskip


\bigskip





\JEL{C63 \and C52 \and G13}
\end{abstract}


\section{Introduction}
\label{s:Introduction}

Crude oil is by far the world's most actively traded commodity.
It is usually traded on exchanges in the form of futures contracts.
The two most important benchmark crudes are West Texas Intermediate (WTI), traded on the NYMEX, and Brent, traded on the ICE.
Recently, the Dubai Mercantile Exchange's (DME) Oman contract has been attracting investors looking for a Middle Eastern sour crude oil benchmark.
In the S\&P Goldman Sachs Commodity Index, WTI has a weight of $24.71 \%$ and Brent a weight of $22.34 \%$, for a combined total of almost half the index.
Another widely quoted index, Jim Rogers' RICI, has weights of $21 \%$ for WTI and $14 \%$ for Brent.
The crude oil derivatives market is also the most liquid commodity derivatives market.
Popular products are European, American, Asian, and calendar spread options on futures contracts.

An important empirical feature of crude oil markets is the absence of seasonality, which is in marked contrast to, say, agricultural commodities markets.
A second empirical feature is stochastic volatility of futures contracts, which is clearly reflected in the oil volatility index (OVX), or ``Oil VIX'', introduced on the CBOE in July 2008.
A third feature is known as the {\it Samuelson effect} \citep{Samuelson1965,BessembinderCoughenourSeguinSmoller1996},
i.e. the empirical observation that a given futures contract increases in volatility as it approaches its maturity date.
Finally, European and American options on futures (usually specified to expire just a couple of days before the underlying futures contract itself)
tend to show a more or less strongly pronounced volatility smile, the shape of which depends on the option's maturity.

European and American options depend on the evolution of just one underlying futures contract.
In contrast to these, calendar spread options have a payoff that is calculated from the difference of two futures contracts with different maturities.
Therefore, a mathematical analysis and evaluation of calendar spread options must be carried out in a framework that models the joint stochastic behaviour of several futures contracts.

In this article, we propose a multi-factor stochastic volatility model for the crude oil futures curve.
Like the popular \citet{ClewlowStrickland1999_1,ClewlowStrickland1999_2} models, the model is futures-based, not spot-based,
which means it can trivially match any given futures curve by accordingly specifying the futures' initial values without ``using up'' any of the other model parameters.
The variance processes are based on the \citet{CoxIngersollRoss1985} and \citet{Heston1993} stochastic variance process.
However, in order to capture the Samuelson effect, we add expiry dependent exponential damping factors.
As in the \citet{Heston1993} model, futures returns and variances are correlated, so that volatility smiles of American and European options observed in the market can be closely matched.
The instantaneous correlation of the returns of two futures contracts is also stochastic in our multi-factor model, since it is calculated from the stochastic variances.

Our first result is the calculation of the joint characteristic function of the log-returns of two futures contracts in analytic form.
Using this function, calendar spread option prices can be obtained via $1$-dimensional Fourier integration as shown by \citet{CaldanaFusai2013}
or the $2$-dimensional Fast Fourier Transform (FFT) algorithm of \citet{HurdZhou2010}.
The fast speed of these algorithms is of great importance when calibrating the model to these products.

Our second result is to give analytical formulas to recover the dependence structure of two futures prices from the joint characteristic function of the model.
The proposed expressions give the dependence structure as the copula function of the two prices or as its copula density.
In many studies, the measure chosen to describe dependence is Pearson's rho, which, however, also depends on the marginal distributions.
In order to completely insulate our analysis from the influence of the marginal distributions,
we carry it out via the copula function produced by the joint characteristic function.
Once we have the copula and its copula density it becomes possible to compute various dependence and concordance measures,
such as Spearman's rho and Kendall's tau, for two futures contracts.


Using these mathematical tools, we carry out an analysis of the dependence of the returns of two futures contracts.
We observe that, for a fixed time-horizon, these returns become {\it less dependent} as the maturity of the second underlying futures contract increases and moves away from that of the first underlying contract.
In analogy to the classic Samuelson {\it volatility} effect, we call this effect the Samuelson {\it correlation} effect.

Copula functions can also be used to give a rigorous definition of the {\it implied correlation} of calendar spread options.
The traditional definition assumes a bivariate Black-Scholes-Merton model for the two underlyings, which assumes in particular that the marginal distributions are log-normal.
In contrast, here, following \citet{Tavin2014}, and using the actual marginal distributions of the model, for a given calendar spread option price
we define the implied correlation as the value of the correlation parameter in the bivariate Gaussian copula that reproduces this price.
Note that implied correlation depends both on the strike and the maturity of the option, phenomena usually referred to as {\it correlation smile/skew/frown} and {\it correlation term structure}.

In an empirical section we calibrate the two-factor version of our model to market data from three different dates.
We show that the model can fit European and American option prices as well as calendar spread option prices very closely.
The model could therefore be used by a price maker in a crude oil market to provide consistent and arbitrage free prices to other market participants.

We conclude this introduction with a survey of previous literature.
A detailed exposition of commodity models is given by \citet{Clark2014}.
One of the most important and still widely used models is the \citet{Black1976} futures model, which is set in the Black-Scholes-Merton framework.
Contracts with different maturities can have different volatilities in this model, but for each contract the volatility is constant.
Therefore, Black's model doesn't capture the Samuelson effect.
European option prices in this model are given by the Black-Scholes-Merton formula, and consequently there is no volatility smile for options with different strikes.
Finally, all contracts are perfectly correlated in this model, since they are driven by the same Brownian motion.

\citet{ClewlowStrickland1999_1,ClewlowStrickland1999_2} propose one-factor and multi-factor models of the entire futures curve with deterministic time-dependent volatility functions.
A popular specification for these functions is with exponential damping factors.
Since this specification still leads to log-normally distributed futures prices, there is no volatility smile or skew in this model.
In the one-factor model, the instantaneous returns of contracts with different maturities are perfectly correlated;
in the multi-factor model, however, these returns are not perfectly, but deterministically correlated.

Stochastic volatility models have been proposed by \citet{Scott1987,Scott1997}, \citet{Wiggins1987}, \citet{HullWhite1987}, \citet{SteinStein1991},
\citet{Heston1993}, \citet{BakshiCaoChen1997}, and \citet{SchoebelZhu1999}, among others.
Extending the \citet{Heston1993} model to multiple factors, \citet{ChristoffersenHestonJacobs2009} show that under certain independence assumptions
it is straightforward to obtain the characteristic function in closed form and calculate European option prices using the Fourier transform.
An important aspect they then proceed to study is the stochastic correlation between the stock return and variance implied by the model.
\citet{DuffiePanSingleton2000} have studied a very general class of jump-diffusions.
The model presented in this paper fits into this framework (in the version extended to time-dependent parameters).

\citet{TrolleSchwartz2009b} introduce a very general two-factor spot based model,
with, in addition, two stochastic volatility factors as well as two stochastic factors for the forward cost-of-carry.
The variance processes are extensions of the CIR/Heston process to a more general mean-reversion specification.
They also give the dynamics of their model in terms of the futures curve.
The main focus of their study is on unspanned stochastic volatility of single-underlying options on futures contracts.


Spread options have been well studied in a two-factor Black-Scholes-Merton framework.
\citet{Margrabe1978} gives an exact formula when the strike $K$ equals zero, and \citet{Kirk1995}, \citet{CarmonaDurrleman2003},
\citet{BjerksundStensland2011} and \citet{VenkatramananAlexander2011} give approximation formulas for any $K$.

\citet{CaldanaFusai2013} have recently proposed a very fast one-dimensional Fourier method that extends the approximation given by \citet{BjerksundStensland2011}
for the Black-Scholes-Merton model to any model for which the joint characteristic function is known.
As in \citet{BjerksundStensland2011}, the method provides a lower bound for the spread option price, but in practice the bound seems to be so close to the actual option price that it can be used as the price itself.

\citet{CarrMadan1999} show how the Fast Fourier Transform (FFT) can be used to price European options with different strikes in one step.
\citet{DempsterHong2002} and \citet{HurdZhou2010} apply the two-dimensional FFT to the pricing of spread options.
Hurd and Zhou's method returns spread option prices at many different strikes (after a re-scaling and interpolation step), in analogy to \citet{CarrMadan1999}, in one inversion step.

The rest of the paper proceeds as follows.
In Section \ref{s:StochasticVolatilityModelForCrudeOil} we define the proposed model and provide the associated joint characteristic function.
Section \ref{s:CalendarSpreadOptionsAndAnalysisOfDependence} deals with spread options and the structure of dependence produced by the model.
Section \ref{s:Calibration} presents an empirical analysis based on different market situations.
Section \ref{s:Conclusion} concludes.


\section{A Model with Stochastic Volatility for Crude Oil Futures}
\label{s:StochasticVolatilityModelForCrudeOil}

\subsection{The Financial Framework and the Model}
\label{ss:FinancialFrameworkAndModel}

We begin by giving a mathematical description of our model under the risk-neutral measure ${\mathbb Q}$.
Let $n \geq 1$ be an integer, and let $B_1, ..., B_{2n}$ be Brownian motions under ${\mathbb Q}$.
Let $T_m$ be the maturity of a given futures contract.
The futures price $F(t, T_m)$ at time $t, 0 \leq t \leq T_m$, is assumed to follow the stochastic differential equation (SDE)
\begin{equation}
\label{FuturesSDE}
dF(t, T_m) = F(t, T_m) \sum_{j=1}^n e^{-\lambda_j(T_m - t)} \sqrt{v_j(t)} dB_j(t), \; F(0, T_m) = F_{m,0} > 0.
\end{equation}
The processes $v_j, j=1, ..., n,$ are CIR/Heston square-root stochastic variance processes assumed to follow the SDE
\begin{equation}
\label{VarianceSDE}
dv_j(t) = \kappa_j \left( \theta_j - v_j(t) \right) dt + \sigma_j \sqrt{v_j(t)} dB_{n+j}(t), \; v_j(0) = v_{j,0} > 0.
\end{equation}
For the correlations, we assume
\begin{equation}
\label{FuturesVarianceCorrelations}
\langle dB_j(t), dB_{n+j}(t) \rangle = \rho_j dt, -1 < \rho_j < 1, j=1, ..., n,
\end{equation}
and that otherwise the Brownian motions $B_j, B_k, k \neq j, j + n,$ are independent of each other.
As we will see, this assumption has as a consequence that the characteristic function factors into $n$ separate expectations.
Note also that the identity
$\det \left( \begin{array}{cc} I_n & B \\ C & I_n \end{array} \right) = \det \left( I_n - B C \right)$
together with Sylvester's criterion can be used to show that the correlation matrix determined by \eqref{FuturesVarianceCorrelations}
is indeed positive definite for any choice of the parameters $\rho_j, j=1, ..., n$.

For fixed $T_m$, the futures log-price $\ln F(t, T_m)$ follows the SDE
\begin{equation}
\label{LogFuturesSDE}
d\ln F(t, T_m) = \sum_{j=1}^n \left( e^{-\lambda_j(T_m - t)} \sqrt{v_j(t)} dB_j(t) - \frac{1}{2} e^{-2\lambda_j(T_m - t)} v_j(t) dt \right),
\; \ln F(0, T_m) = \ln F_{m,0}.
\end{equation}
Integrating \eqref{LogFuturesSDE} from time $0$ up to a time $T, T \leq T_m$, gives
\begin{equation}
\label{LogFuturesIntegratedSDE}
\ln F(T, T_m) -  \ln F(0, T_m) = \sum_{j=1}^n \int_0^T e^{-\lambda_j(T_m - t)} \sqrt{v_j(t)} dB_j(t) - \frac{1}{2} \sum_{j=1}^n \int_0^T e^{-2\lambda_j(T_m - t)} v_j(t) dt.
\end{equation}

We define the log-return between times $0$ and $T$ of a futures contract with maturity $T_m$ as
$$
X_m(T) := \ln \left( \frac{F(T, T_m)}{F(0, T_m)} \right).
$$
In the following, the joint characteristic function $\phi$ of two log-returns $X_1(T), X_2(T)$ will play an important role.
For $u = (u_1, u_2) \in {\mathbb C}^2$, $\phi$ is given by
\begin{equation}
\label{jointCharacteristicFunction_returns}
\phi(u)
= \phi(u; T, T_1, T_2)
= {\mathbb{E^Q}} \left[ \exp \left( i \sum_{k=1}^2 u_k X_k(T) \right) \right].
\end{equation}
The joint characteristic function $\Phi$ of the futures log-prices $\ln F(T, T_1), \ln F(T, T_2)$ is then given by
\begin{equation}
\label{jointCharacteristicFunction_prices}
\Phi(u) = \exp \left( i \sum_{k=1}^2 u_k \ln F(0, T_k) \right) \cdot \phi(u).
\end{equation}
Note that futures prices in our model are not mean-reverting,
and that the log-price $\ln F(t, T_m)$ at time $t$ and the log-return $\ln F(T, T_m) - \ln F(t, T_m)$ are independent random variables.

In the following proposition, we show how the joint characteristic function $\phi$, and therefore also the single characteristic function $\phi_1$,
is given by a system of two ordinary differential equations (ODE).
\begin{proposition}
\label{Prop:JointCharacteristicFunction}
The joint characteristic function $\phi$ at time $T \leq T_1, T_2$ for the log-returns $X_1(T), X_2(T)$ of two futures contracts with maturities $T_1, T_2$ is given by
\begin{align*}
\phi(u) &= \phi(u; T, T_1, T_2)
\\
&=
\prod_{j=1}^n
\exp \left( i \frac{\rho_j}{\sigma_j} \left\{ \frac{\kappa_j \theta_j}{\lambda_j} (f_{j,1}(u,0) - f_{j,1}(u,T)) - f_{j,1}(u,0) v_j(0) \right\} \right)
\exp \left( A_j(0,T) v_j(0) + B_j(0,T) \right),
\end{align*}
where
\begin{align*}
f_{j,1}(u,t) &= \sum_{k=1}^2 u_k e^{- \lambda_j(T_k - t)}, \quad f_{j,2}(u,t) = \sum_{k=1}^2 u_k e^{-2\lambda_j(T_k - t)},
\\
q_j(u,t)   &= i \rho_j \frac{\kappa_j - \lambda_j}{\sigma_j} f_{j,1}(u,t) - \frac{1}{2} (1 - \rho_j^2) f_{j,1}^2(u,t) - \frac{1}{2} i f_{j,2}(u,t),
\end{align*}
and the functions $A_j: (t,T) \mapsto A_j(t,T)$ and $B_j: (t,T) \mapsto B_j(t,T)$ satisfy the two differential equations
\begin{align*}
\frac{\partial A_j}{\partial t} - \kappa_j A_j + \frac{1}{2} \sigma_j^2 A_j^2 + q_j &= 0,
\\
\frac{\partial B_j}{\partial t} + \kappa_j \theta_j A_j &= 0,
\end{align*}
with $A_j(T,T) = i \frac{\rho_j}{\sigma_j} f_{j,1}(u,T), \; B_j(T,T) = 0.$

The single characteristic function $\phi_1$ at time $T \leq T_1$ for the log-return $X_1(T)$ of a futures contract with maturity $T_1$ is given by setting $u_2 = 0$
in the joint characteristic function.
\end{proposition}
%
The statement regarding the single characteristic function immediately follows from the definition of the joint characteristic function.
The joint characteristic function is calculated in appendix \ref{a1:Proofs}.

In the next proposition, we show how this ODE system can be solved analytically.
A closed form expression for $A_j$ is found thanks to a computer algebra software and $B_j$ is then proportional to the integral of $A_j$ on $[0,T]$.

\begin{proposition}
\label{Prop:AnalyticSolution_A}
Dropping the references to $j$, the function $A: (t,T) \mapsto A(t,T)$ is given in closed form as
\begin{align*}
A(t,T) & = \frac{1}{\sqrt{2} z \sigma} \cdot
\frac{\left( (M^-(t)-M^+(t)) X_0 + X_1 U^+(t) \right) C_1 - 2 \left( M^+(t) X_0 - X_1 U^+(t) \right)(C_3-i C_2)e^{\lambda t}}{M^+(t) X_0 - X_1 U^+(t)} \\[10pt]
&\quad + \frac{1}{\sigma^2} \cdot \frac{\left( (\kappa - \lambda)M^+(t) + (\kappa + \lambda)M^-(t) \right) X_0 - \left( (\kappa - \lambda) U^+(t) - 2 \lambda U^-(t) \right) X_1}{M^+(t) X_0 - X_1 U^+(t)},
\end{align*}
with $z = \sqrt{C_2+i C_3}$ and $C_1, C_2, C_3$ constants with respect to $t$, defined as
\begin{equation*}
C_1 = \rho \frac{\kappa - \lambda}{\sigma} \sum_{k=1}^2 u_k e^{- \lambda T_k}, \quad
C_2 = - \frac{1}{2} (1 - \rho^2) \left(\sum_{k=1}^2 u_k e^{- \lambda T_k} \right)^2, \quad
C_3 = - \frac{1}{2} \sum_{k=1}^2 u_k e^{-2\lambda T_k},
\end{equation*}
\begin{align*}
X_0 &= 2 Y U^+(T) + 4 z \lambda U^-(T), \\
X_1 &= 2 Y M^+(T) - 2 \left( z (\lambda + \kappa) + \sigma \sqrt{2} \frac{C_1}{2} \right) M^-(T), \\
Y &= \sigma \sqrt{2} \left(\frac{C_1}{2} -i e^{\lambda T} C_2  + e^{\lambda T} C_3\right) - z \left(\kappa - \lambda - i \rho f_1(T) \sigma \right), \\
M^{\pm}(t) &= M \left(\frac {\kappa z - \frac{\sigma \sqrt{2}}{2} C_1}{2 z \lambda} \pm \frac{1}{2}, \frac{\kappa + \lambda}{\lambda}, \frac{\sigma \sqrt{2}}{\lambda} i z {e^{\lambda \, t}} \right), \\
U^{\pm}(t) &= U \left(\frac {\kappa z - \frac{\sigma \sqrt{2}}{2} C_1}{2 z \lambda} \pm \frac{1}{2}, \frac{\kappa + \lambda}{\lambda}, \frac{\sigma \sqrt{2}}{\lambda} i z {e^{\lambda \, t}} \right).
\end{align*}
The functions $M$ and $U$ are the confluent hypergeometric functions.
\end{proposition}

$M$ and $U$ are usually referred to as Kummer's functions as they solve Kummer's equation \citep{Kummer1836,Tricomi1955}.
The function $M$ is also known as ${}_1 F_1$, and the function $U$ as Tricomi's function.
Given, $a,b,z \in \mathbb{C}$, Kummer's equation is
\begin{equation}
z \frac{\partial^2 w}{\partial z^2} + (b-z)\frac{\partial w}{\partial z} - a w = 0.
\end{equation}
A way to obtain $M(a,b,z)$ is by means of a series expansion
\begin{equation}
M(a,b,z) = 1 + \sum^{\infty}_{n=1}{\frac{z^n \prod^{n}_{j=1}{(a+j-1)}}{n! \prod^{n}_{j=1}{(b+j-1)}}}.
\end{equation}
And $U(a,b,z)$ is obtained from $M$ as
\begin{equation}
U(a,b,z) = \frac{\pi}{\sin{(\pi b)}} \left(\frac{M(a,b,z)}{\Gamma(1+a-b)\Gamma(b)} - z^{1-b} \frac{M(1+a-b,2-b,z)}{\Gamma(a)\Gamma(2-b)} \right),
\end{equation}
where $\Gamma$ denotes the Gamma function extended to the complex plane.
These results and additional properties of Kummer's functions (e.g. integral representations) can be found in Chap. 13 of \citet{AbramovitzStegun1972}.
A detailed analysis of how to implement Kummer's functions is given by \citet{Pearson2009}.
A suitable way to implement the complex Gamma function is the \citet{Lanczos1964} approximation.

As has already been mentioned, the model introduced above is an extension of the one-factor \citet{ClewlowStrickland1999_1} and multi-factor \citet{ClewlowStrickland1999_2}
models to stochastic volatility.
Since these models are useful benchmarks, we give a description of them and calculate their joint characteristic function.
In the risk-neutral measure ${\mathbb Q}$, the futures price $F(t, T_m)$ is modelled with deterministic time-dependent volatility functions $\hat{\sigma}_j(t,T_m)$:
\begin{equation}
\label{ClewlowStrickland_SDE}
dF(t, T_m) = F(t, T_m) \sum_{j=1}^n \hat{\sigma}_j(t,T_m) dB_j(t),
\end{equation}
where $B_1, ..., B_n$ are independent Brownian motions.
A popular specification for the volatility functions is
\begin{equation}
\label{ClewlowStrickland_VolatilityFunction}
\hat{\sigma}_j(t,T_m) := e^{-\lambda_j(T_m - t)} \sigma_j
\end{equation}
for fixed parameters $\sigma_j, \lambda_j \geq 0$,
so that the volatility of a contract a long time away from its maturity is damped by the exponential factor(s).

The marginal and joint distributions of futures prices are log-normal in the Clewlow-Strickland models.
Nevertheless, it can be very useful to know the single and joint characteristic functions as well for testing and benchmarking purposes.
The next proposition gives closed-form solutions for them.
\begin{proposition}
\label{Prop:JointCharacteristicFunction_ClewlowStrickland}
In the Clewlow and Strickland model defined by \eqref{ClewlowStrickland_SDE} and \eqref{ClewlowStrickland_VolatilityFunction},
the joint characteristic function $\phi$ at time $T \leq T_1, T_2$ for the log-returns $X_1(T), X_2(T)$ of two futures contracts with maturities $T_1, T_2$ is given by
\begin{align*}
\phi(u) &= \phi(u; T, T_1, T_2)
\\
&= \prod_{j=1}^n
\exp \left( -\frac{\sigma_j^2}{4 \lambda_j} (e^{2 \lambda_j T} - 1)
\left\{
i (u_1 e^{-2 \lambda_j T_1} + u_2 e^{-2 \lambda_j T_2}) + (u_1 e^{-\lambda_j T_1} + u_2 e^{-\lambda_j T_2})^2
\right\} \right).
\end{align*}
The single characteristic function $\phi_1$ at time $T \leq T_1$ for the log-return of a futures contract with maturity $T_1$ is given by setting $u_2 = 0$
in the joint characteristic function.
\end{proposition}
We prove this result in appendix \ref{a1:Proofs}.

This result can also be used to add non-stochastic volatility factors to the model by multiplying the joint characteristic function of
Proposition \ref{Prop:JointCharacteristicFunction} with one or more factors from Proposition \ref{Prop:JointCharacteristicFunction_ClewlowStrickland}.
Since each ``Clewlow-Strickland'' factor depends on only two parameters $\lambda_j$ and $\sigma_j$, it does not add a significant burden to the
calibration to market data, while allowing for increased flexibility when fitting the model to the observed volatility term structure.

\subsection{Pricing Vanilla Options}
\label{ss:PricingVanillaOptions}

European options on futures contracts can be priced using the Fourier inversion technique as described in \citet{Heston1993} and \citet{BakshiMadan2000},
or the FFT algorithm of \citet{CarrMadan1999}. Alternatively, they can be priced by Monte Carlo simulation using discretizations of
\eqref{FuturesSDE} (Euler scheme) or \eqref{LogFuturesSDE} (Log-Euler scheme) and of \eqref{VarianceSDE}.

Let $K$ denote the strike and $T$ the maturity of a European call option on a futures contract $F$ with maturity $T_m \geq T$,
and let the single characteristic function $\Phi_1$ of the futures log-price $\ln F(T, T_m)$ be given by
$\Phi_1(u) = e^{i u \ln F(0, T_m)} \phi_1(u)$.
In the general formulation of \citet{BakshiMadan2000}, the numbers
\begin{align}
\label{Pi1}
\Pi_1 &:= \frac{1}{2} + \frac{1}{\pi} \int_0^\infty \Re \left[ \frac{e^{-i u \ln K} \phi_1(u - i)}{i u \phi_1(-i)} \right] du,
\\
\label{Pi2}
\Pi_2 &:= \frac{1}{2} + \frac{1}{\pi} \int_0^\infty \Re \left[ \frac{e^{-i u \ln K} \phi_1(u)}{i u} \right] du,
\end{align}
represent the probabilities of $F$ finishing in-the-money at time $T$ in case the futures $F$ itself or a risk-free bond is used as num\'{e}raire, respectively.
The price $C$ of a European call option is then obtained with the formula
\begin{equation}
C = e^{-rT} \left( F(0, T_1) \Pi_1 - K \Pi_2 \right).
\end{equation}
European put options can be priced via put-call parity $C - P = e^{-rT} \left( F(0, T_1) - K \right)$.

American call and put options can be evaluated via Monte-Carlo simulation using the method of \citet{LongstaffSchwartz2001}.
Alternatively, the early exercise premium can be approximated with the formula of \citet{BaroneAdesiWhaley1987}.
\citet{TrolleSchwartz2009b} (Appendix B) address the issue of estimating European prices from American prices.

A typical WTI volatility surface displays high implied volatilities at the short end and low implied volatilities at the long end.
This is in line with the Samuelson effect.
Furthermore, there is usually a strongly pronounced smile at the short end, and a weak smile at the long end.

\subsection{Stochastic Correlation in the Multi-Factor Model}
\label{ss:StochasticCorrelation}

We will show in this section that if we specify our model with two or more volatility factors, then the returns of two given futures contracts are stochastically correlated,
which is a realistic and important feature.

Define
\begin{equation}
\label{Vij}
V_{ij}(t) := \langle \frac{dF(t, T_i)}{F(t, T_i)}, \frac{dF(t, T_j)}{F(t, T_j)} \rangle / dt.
\end{equation}
Then the instantaneous correlation $\rho(t)$ at time $t$ is given by:
\begin{equation}
\label{rho_tT1T2}
\rho(t) = \frac{V_{12}(t)}{\sqrt{V_{11}(t)} \sqrt{V_{22}(t)}}.
\end{equation}

Let us begin with an examination of the $1$-factor model,
in which futures returns follow the SDE
\begin{equation}
\label{FuturesSDE_1FactorModel}
\frac{dF(t, T_m)}{F(t, T_m)} = e^{-\lambda_1(T_m - t)} \sqrt{v_1(t)} dB_1(t)
\end{equation}
and the variance process follows the SDE
\begin{equation}
\label{VarianceSDE_1FactorModel}
dv_1(t) = \kappa_1 \left( \theta_1 - v_1(t) \right) dt + \sigma_1 \sqrt{v_1(t)} dB_2(t).
\end{equation}
The correlation is given by $\langle dB_1(t), dB_2(t) \rangle = \rho_1 dt$.
Inserting \eqref{FuturesSDE_1FactorModel} into \eqref{Vij} gives for the instantaneous covariance
\begin{equation}
\label{Vij_1FactorModel}
V_{12}(t) = e^{-\lambda_1(T_1 + T_2 - 2t)} v_1(t).
\end{equation}
\cite{CoxIngersollRoss1985} show that the random variable $v_1(t)$ follows a non-central $\chi^2$-distribution.
It is easy to see that the instantaneous correlation \eqref{rho_tT1T2} is always equal to one in the $1$-factor model:
\begin{equation}
\label{rho_tT1T2_1FactorModel}
\rho(t) = \frac{e^{-\lambda_1(T_1 + T_2 - 2t)} v_1(t)}{\sqrt{e^{-2\lambda_1(T_1 - t)} v_1(t)} \sqrt{e^{-2\lambda_1(T_2 - t)} v_1(t)}} = 1.
\end{equation}
Finally, the terminal covariance is given by
\begin{equation*}
\int_0^T \langle \frac{dF(t, T_1)}{F(t, T_1)}, \frac{dF(t, T_2)}{F(t, T_2)} \rangle = \int_0^T e^{-\lambda_1(T_1 + T_2 - 2t)} v_1(t) dt.
\end{equation*}
What can we say about its distribution?

\citet{AlbaneseLawi2005} consider the Laplace transform of such integrals (see also \citet{HurdKuznetsov2008})
in general, and in particular for the CIR/Heston-process: 
\begin{equation*}
L_{T-t}(X_t, \vartheta) = E^P \left[ e^{-\vartheta \int_t^T \phi(X_s) ds} q(X_T) | \mathcal{F}_t \right]
\end{equation*}
where $t \leq T, \vartheta \in \mathbb{C}$ and $\vartheta, q: \mathbb{R} \to \mathbb{R}$ are two Borel functions.
However, they come to the conclusion in Corollary 3, eq. (50), that the Laplace transform of the integral in our case,
which includes an exponential factor, is not computable in closed form.

Next, we examine the $2$-factor model, in which futures returns follow the SDE
\begin{equation}
\label{FuturesSDE_2FactorModel}
\frac{dF(t, T_m)}{F(t, T_m)} = e^{-\lambda_1(T_m - t)} \sqrt{v_1(t)} dB_1(t) + e^{-\lambda_2(T_m - t)} \sqrt{v_2(t)} dB_2(t).
\end{equation}
and the two variance processes follow the SDEs
\begin{align}
\label{VarianceSDE1_2FactorModel}
dv_1(t) &= \kappa_1 \left( \theta_1 - v_1(t) \right) dt + \sigma_1 \sqrt{v_1(t)} dB_3(t),
\\
\label{VarianceSDE2_2FactorModel}
dv_2(t) &= \kappa_2 \left( \theta_2 - v_2(t) \right) dt + \sigma_2 \sqrt{v_2(t)} dB_4(t).
\end{align}
The correlations are given by $\langle dB_1(t), dB_3(t) \rangle = \rho_1 dt$, $\langle dB_2(t), dB_4(t) \rangle = \rho_2 dt$,
and all other correlations are zero.

Inserting \eqref{FuturesSDE_2FactorModel} into \eqref{Vij} gives for the instantaneous covariance
\begin{equation}
\label{Vij_2FactorModel}
V_{12}(t) = e^{-\lambda_1(T_1 + T_2 - 2t)} v_1(t) + e^{-\lambda_2(T_1 + T_2 - 2t)} v_2(t).
\end{equation}
In contrast to the $1$-factor model, the instantaneous correlation $\rho(t)$ in the $2$-factor model is now stochastic.
The same holds of course for the general multi-factor model with $n \geq 2$.

What can we say about the distribution of $\rho(t)$?
It follows from the definition \eqref{rho_tT1T2} that $0 < \rho(t) \leq 1$, so that the returns of the two futures contracts are always positively correlated.
To get some more insight, we consider the $2$-factor model in a numerical example.
The parameters of the model have been chosen for illustrative purposes and are given in Table \ref{tab:ModelParameters}:
the first factor is more volatile than the second one, and it also decays more slowly.

\begin{table}[htbp]
  \centering
  \caption{Model Parameters}
    \begin{tabular}{cr|cr}
    \addlinespace
    \toprule
    $\kappa_1$  & $1.00$ & $\kappa_2$  & $1.00$ \\
    $\theta_1$  & $0.16$ & $\theta_2$  & $0.09$ \\
    $\rho_1$    & $0.00$ & $\rho_2$    & $0.00$ \\
    $\sigma_1$  & $0.25$ & $\sigma_2$  & $0.20$ \\
    $v_1(0)$    & $0.16$ & $v_2(0)$    & $0.09$ \\
    $\lambda_1$ & $0.10$ & $\lambda_2$ & $2.00$ \\
    \bottomrule
    \end{tabular}
  \label{tab:ModelParameters}
\end{table}

For two contracts with maturities $T_1 = 1$ and $T_2 = 2$ years, respectively, we plot the empirical density function of $\rho(t)$ in Figure \ref{Fig:Correlations}.

\begin{figure}[ht]
\centering
\includegraphics[width=.75\textwidth]{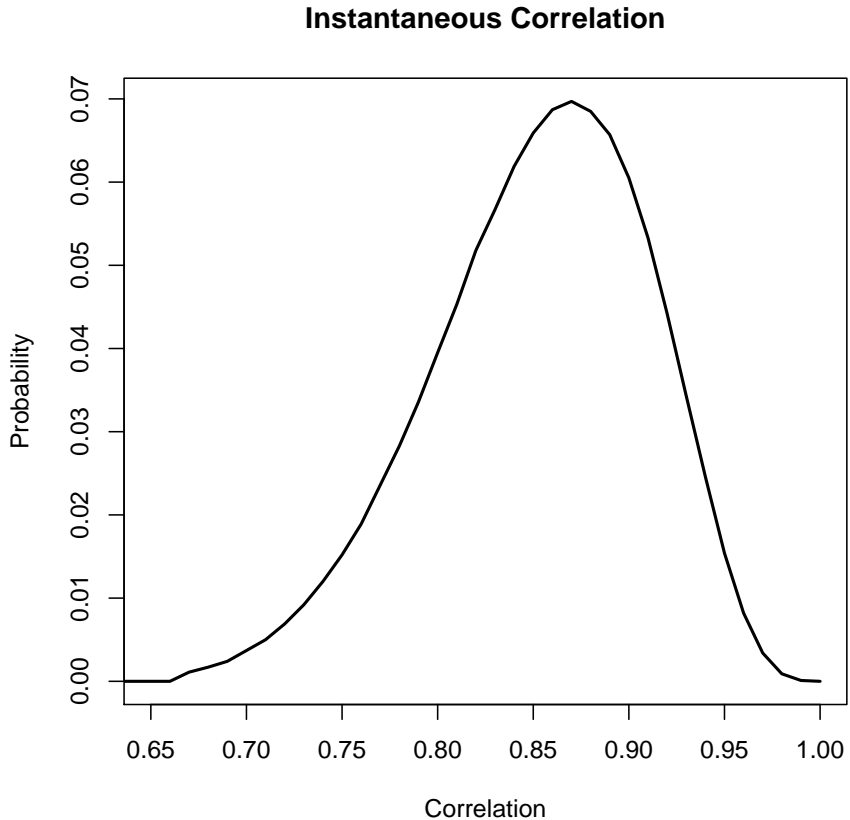}
\caption{Empirical probabilities of the instantaneous correlation $\rho(1; 1, 2)$}
\label{Fig:Correlations}
\end{figure}

These plotted empirical probabilities were obtained by sampling \eqref{rho_tT1T2} one million times in a Monte Carlo simulation.
The empirical mean is $\overline{\rho} = 0.8575$.
In case both stochastic volatilities are made deterministic by setting $\sigma_1 = \sigma_2 = 0$
in equations \eqref{VarianceSDE1_2FactorModel} and \eqref{VarianceSDE2_2FactorModel},
the empirical mean is $\overline{\rho}^0 = 0.8619$, which is in excellent agreement with the
deterministic instantaneous correlation of a corresponding $2$-factor Clewlow-Strickland model with volatility functions
$\sigma_j(t,T_m) = e^{-\lambda_j(T_m - t)} \hat{\sigma}_j, j=1,2,$
with $\lambda_1 = 0.10, \lambda_2 = 2.00, \hat{\sigma}_1 = \sqrt{\theta_1} = 0.40, \hat{\sigma}_2 = \sqrt{\theta_2} = 0.30.$


\section{Calendar Spread Options and Analysis of Dependence}
\label{s:CalendarSpreadOptionsAndAnalysisOfDependence}

In this Section we review the definition, functioning and pricing of calendar spread options.
We then review the notion of implied correlation associated to a price of calendar spread option.
Finally we introduce analytic results to obtain the copula function and the copula density produced by a model defined by means of its joint characteristic function.

\subsection{Calendar Spread Options written on WTI futures}
\label{ss:CalendarSpreadOptionsWrittenOnWTIFutures}

Calendar spread options (CSO) are very popular options in commodities markets.
There are two types of these options: calendar spread calls (CSC) and calendar spread puts (CSP).
Like spread options in equities derivatives markets, their payoff depends on the price difference of two underlying assets.
A call spread option on two equity shares $S_1$ and $S_2$ gives the holder, at time $T$, the payoff $\max \left( S_1(T) - S_2(T) - K, 0 \right),$
and a put the payoff $\max \left( K - \left( S_1(T) - S_2(T) \right), 0 \right).$
In the case of calendar spread options, the two underlyings are two futures contracts on the same commodity, but with different maturities $T_1$ and $T_2$.
Along with the volatilities, the dependence between the two contracts has a large influence on the option's price.
Note that CSOs written on commodities futures should not be confused with the well-known options strategy named calendar spread.
This strategy involves two vanilla options (one bought and one sold) with different maturities, whereas the CSO is a single option.
Examples of CSOs are the NYMEX calendar spread options on crude oil (WTI).
A WTI CSC (CSP) represents an option to assume a long (short) position in the first expiring futures contract in the spread and a short (long) position in the second contract.
There are also so-called {\it financial} CSOs traded on the NYMEX, which are cash settled.
For pricing purposes we will not distinguish between these two settlement types in this paper.

There is usually very good liquidity on $1$-month spreads (for which $T_2 - T_1 = 1$ month), whereas options on $2, 3, 6$ and $12$-month spreads are less liquid.
The NYMEX CSO on $1$-month WTI spreads can be accessed in Bloomberg using the ticker \texttt{WA}.

Let two futures maturities $T_1, T_2$, an option maturity $T$, and a strike $K$ (which is allowed to be negative) be fixed.
Then the payoffs of calendar spread call and put options, $CSC$ and $CSP$, are respectively given by
\begin{equation}
\label{CalendarSpreadCallPayoff}
CSC(T) = \left( F(T, T_1) - F(T, T_2) - K \right)^+,
\end{equation}
\begin{equation}
\label{CalendarSpreadPutPayoff}
CSP(T) = \left( K - \left( F(T, T_1) - F(T, T_2) \right) \right)^+ .
\end{equation}

To evaluate such options with a pricing model, the discounted expectation of the payoff must be calculated in the risk-neutral measure.
Assuming a continuously-compounded risk-free interest rate $r$, we have at time $t_0 = 0$:
\begin{align}
CSC(0, T, T_1, T_2, K) &= e^{-rT} \mathbb{E}_0 \left[ \left( F(T, T_1) - F(T, T_2) - K \right)^+ \right],
\\
CSP(0, T, T_1, T_2, K) &= e^{-rT} \mathbb{E}_0 \left[ \left( K - (F(T, T_1) - F(T, T_2)) \right)^+ \right].
\end{align}

Note that there is a model-independent put-call parity for calendar spread options:
\begin{equation}
\label{CalendarSpreadPutCallParity}
CSC(0) - CSP(0) = e^{-rT} \left( F(0, T_1) - F(0, T_2) - K \right).
\end{equation}

Apart from Monte-Carlo simulation (where simulation of the CIR/Heston process is well-understood), we are aware of three efficient methods to price spread options.
The first two are suitable when the joint characteristic function is available.
The third one is more direct but needs the marginals and joint distribution function of the underlying futures.

The formula of \citet{BjerksundStensland2011} for a joint Black-Scholes model is generalized by \citet{CaldanaFusai2013} to models for which the joint characteristic function is known.
Strictly speaking, these methods give a lower bound for the spread option price.
However, our tests lead us to agree with the above authors that this lower bound is very close to the actual price (typically the first three digits after the comma are the same),
and we therefore regard this lower bound as the spread option's price itself. Furthermore, in case $K = 0$ the formula is exact (exchange option case).

Let $\Phi_T(u) = \Phi(u)$ be the joint characteristic function of the logarithms $\ln F(T, T_1), \ln F(T, T_2)$ of two futures prices
as given in equation \eqref{jointCharacteristicFunction_prices}.
Following \citet{CaldanaFusai2013}, the price of the calendar spread option call with maturity $T$ and strike $K$ is given in terms of a Fourier inversion formula as
\begin{equation}
\label{CaldanaFusaiFormula}
CSC(0,K,T,T_1,T_2) = \left(\frac{e^{-\delta k - rT}}{\pi} \int_0^{+\infty} e^{-i \gamma k} \Psi_T(\gamma; \delta, \alpha) d\gamma \right)^+,
\end{equation}
where
\begin{align*}
\Psi_T(\gamma; \delta, \alpha) & = \frac{e^{i (\gamma - i \delta) \ln(\Phi_T(0, -i \alpha))}}{i (\gamma - i \delta)} \\
\cdot & \left[ \Phi_T \left( (\gamma - i \delta) - i, -\alpha (\gamma - i \delta) \right) - \Phi_T \left( \gamma - i \delta, -\alpha (\gamma - i \delta) - i \right) - K \Phi_T \left( \gamma - i \delta, -\alpha (\gamma - i \delta) \right) \right]
\end{align*}
and
\begin{equation*}
\alpha = \frac{F(0,T_2)}{F(0,T_2) + K}, \qquad \; k = \ln(F(0,T_2) + K).
\end{equation*}
The parameter $\delta$ controls an exponential decay term as in \citet{CarrMadan1999}.
We found a value of $\delta = 1$ to perform well in numerical applications.
This method appears to be the most suitable to our model and setup.

An alternative method that also works with the joint characteristic function of the log-returns has been proposed by \citet{HurdZhou2010}.
In their paper, the transform of the calendar spread payoff function with a strike of $K = 1$ is analytically calculated.
The price of the corresponding option is then deduced from this analytical result.
Let $x_m (T) := \ln F(T, T_m)$ denote the time $T$ futures log-price.
Following \citet{HurdZhou2010}, the calendar spread call option with maturity $T$, underlying futures maturities $T_1, T_2$, and strike $K = 1$,
can be priced as:
\begin{equation}
\label{HurdZhouFormula}
CSC(0,K=1,T,T_1,T_2) = \frac{e^{-rT}}{4 \pi^2} \iint_{\mathbb{R}^2 + i \epsilon} \phi(u; T, T_1, T_2) \hat{P}(u) d^2 u,
\end{equation}
where
\begin{equation*}
\hat{P}(u) = \frac{\Gamma(i(u_1 + u_2) - 1) \Gamma(-i u_2)}{\Gamma(i u_1 + 1)},
\end{equation*}
and $\Gamma$ is the complex gamma function defined for $\Re(z) > 0$ by the integral $\Gamma(z) = \int_0^{\infty} e^{-t} t^{z-1} dt$.
The double integral  in \eqref{HurdZhouFormula} is evaluated numerically using the two-dimensional Fast Fourier Transform ($2d$ FFT).
The algorithm returns a whole matrix of option prices at different values of $x_1 = \ln F(0, T_1)$ and $x_2 = \ln F(0, T_2)$.
Options with other strikes ($K \neq 1$) are then evaluated by re-scaling and interpolation, if necessary, using the same matrix.

Methods working with distribution functions instead of characteristic functions are also available to price calendar spread options.
In this category of methods, the most direct approach is to evaluate a double integral of the payoff function times the joint density of the two underlying futures contracts.
However, following \citet{Tavin2014}, we can write calendar spread option prices as {\it single} integrals over the marginal and joint distribution functions.
%
The calendar spread call and put option prices are given, at $t=0$ and for $K \geq 0$, by
\begin{align}
CSC(0,K, T, T_1, T_2) & = \int_0^{+\infty} \left( G_2(x, T, T_2) - G(x, x+K, T, T_1, T_2) \right) dx,
\label{CalendarSpreadCallOptionFormula} \\
CSP(0, T, T_1, T_2, K) & = \int_0^{+\infty} \left( G_1(x + K, T, T_1) - G(x, x+K, T, T_1, T_2) \right) dx,
\label{CalendarSpreadPutOptionFormula}
\end{align}
where $G_1$ and $G_2$ are the marginal distribution functions of $X_1$ and $X_2$, respectively, and $G$ is their joint distribution function.
The case $K < 0$ is treated as a calendar spread option written on the reverse spread $F(T, T_2) - F(T, T_1)$ with the opposite strike $-K$.
A proof for spread options that can easily be adapted to our case is given in \citet{Tavin2014}.

In our model, the distribution functions involved in (\ref{CalendarSpreadCallOptionFormula}) and (\ref{CalendarSpreadPutOptionFormula}) are not readily available.
However, it is possible to calculate $G_1$ and $G_2$ from the joint characteristic function $\phi$ of ($X_1,X_2$) using direct inversion formulas given by
\begin{align}
G_1(x,T,T_1) & = \frac{e^{a x}}{2 \pi}\int^{+\infty}_{-\infty}{e^{-iu x}\frac{\phi(u+ia,0,T,T_1,T_2)}{a-iu}du}, \label{MarginalCDF1Formula} \\
G_2(x,T,T_2) & = \frac{e^{a x}}{2 \pi}\int^{+\infty}_{-\infty}{e^{-iu x}\frac{\phi(0,u+ia,T,T_1,T_2)}{a-iu}du}, \label{MarginalCDF2Formula}
\end{align}
with a proper choice of the smoothing parameter $a>0$. We find that $a=3$ works well in our applications.
A detailed proof of these inversion results can be found in \citet{LeCourtoisWalter2014}.
The joint distribution function $G$ can be recovered in a similar way using a direct two-dimensional inversion formula.
\begin{lemma}
\label{Lemma:JointCDFFormula}
\begin{equation}
G(x_{1},x_{2},T,T_1,T_2)
= \frac{e^{a_1 x_{1}+a_2x_{2}}}{4 \pi^2}\int^{+\infty}_{-\infty}\int^{+\infty}_{-\infty}{e^{-i(u_1 x_{1}+u_2 x_{2})}\frac{\phi(u_1+ia_1,u_2+ia_2,T,T_1,T_2)}{(a_1-iu_1)(a_2-iu_2)}du_1 du_2}.
\label{JointCDFFormula}
\end{equation}
\end{lemma}
The proof of this expression follows along the same lines as the one for the univariate case given by \citet{LeCourtoisWalter2014}.
It is given in appendix \ref{a1:Proofs}.

Here again, a proper choice of the smoothing parameters $a_1,a_2>0$ is needed; we find that $a_1=a_2=3$ works well in our applications.
Inversion formulas (\ref{MarginalCDF1Formula}),(\ref{MarginalCDF2Formula}) and (\ref{JointCDFFormula}) are suitable for the use of FFT methods in one and two dimensions.

\subsection{Implied Correlation for Calendar Spread Options}
\label{ss:ImpliedCorrelationForCalendarSpreadOptions}

Given the observed price of a calendar spread option it is possible to extract an implied quantity reflecting the implied level of dependence embedded in the given price.
This quantity is named implied correlation and can be defined as the parameter of the Gaussian copula that reproduces the observed price.
This definition has the advantage to be well defined in the sense that implied correlation exists as soon as the observed price is free of arbitrage.
This copula-based definition also has the advantage to disentangle the impact of marginals from the dependence structure on the price of a calendar spread option.
This notion of implied correlation is introduced and detailed in \citet{Tavin2014}.

For $\rho \in [-1, 1]$, we denote by $C^{G}_{\rho}$ the bivariate Gaussian copula with parameter $\rho$.
With special cases, for $(u_{1},u_{2})\in [0,1]^{2}$, $C^{G}_{\rho=+1}\left(u_{1},u_{2} \right) = C^{+}\left(u_{1},u_{2} \right)$ and $C^{G}_{\rho=-1}\left(u_{1},u_{2} \right) = C^{-}\left(u_{1},u_{2} \right)$,
where $C^{+}$ and $C^{-}$ are the usual upper and lower Fr\'{e}chet-Hoeffding bounds.
For definitions and general theory about copula functions we refer to \citet{Nelsen2006} and \citet{MaiScherer2012}.

When the chosen dependence structure is given by a Gaussian copula with correlation parameter $\rho$, and its marginals are $G_{1}$ and $G_{2}$,
the price of the strike $K$ calendar spread call option is denoted by $CSC^G$ and is given by, for $\rho\in ]-1,+1[$,
\begin{equation}
CSC^{G}(0,T,T_1,T_2,K,\rho) = e^{-rT} \int_{0}^{+\infty}{\left( G_{1}(x,T) - C^{G}_{\rho}\left(G_{1}(x,T),G_{2}(x+K,T) \right)\right) dx},
\label{spreadopt2}
\end{equation}
and, for $\rho = \pm 1$
\begin{align*}
CSC^{G}(0,T,T_1,T_2,K,\rho=+1) &= CSC^{+}_{0}(K) = e^{-rT} \int_{0}^{1}{\left(G^{-1}_{2}(u,T)-G^{-1}_{1}(u,T)-K\right)^{+} du}, \\
CSC^{G}(0,T,T_1,T_2,K,\rho=-1) &= CSC^{-}_{0}(K) = e^{-rT} \int_{0}^{1}{\left(G^{-1}_{2}(u,T)-G^{-1}_{1}(1-u,T)-K\right)^{+} du},
\end{align*}
where $CSC^{+}$ and $CSC^{-}$ denote the prices obtained for the calendar spread option when the chosen dependence structures are respectively $C^{+}$ and $C^{-}$.

When a calendar spread option quote is observed, it is possible to extract the implied correlation $\rho^*$ as the value of the correlation parameter to be used in (\ref{spreadopt2}) such that it reproduces the given price.
The retained notion of implied correlation does not rely on a bivariate extension of the Black-Scholes-Merton model.
Instead, it corresponds to a bivariate Gaussian copula that pairs the true underlying marginals so that the implied correlation $\rho^*$ does not imbed errors made on the marginals when using a joint log-normal model.

Let $CSC^{\text{obs}}\left(0,K,T,T_1,T_2\right)$ be an observed arbitrage-free price for the calendar spread option call written on $F_1$ and $F_2$, with maturity $T$ and strike $K$.
The associated implied correlation $\rho^*$ exists and is unique.
It can be obtained by solving numerically the equation in $\rho$ that is written
\begin{equation}
\label{impcorrdef}
CSC^{G}(0,T,T_1,T_2,K,\rho) = CSC^{\text{obs}}\left(0,K,T,T_1,T_2\right), \text{ for } \rho\in[-1,+1].
\end{equation}

For proofs and more details on the computation of implied correlation we refer to \citet{Tavin2014}.
Note that, here, the notion of implied correlation has been reviewed with calendar spread calls.
The definition can be given identically for a put option as, by put-call parity, the value is the same as for the calendar spread call option with same characteristics.
On the WTI calendar spread options market, implied correlation depends on both the strike and the maturity of options.
By analogy with implied volatility, these phenomena are referred to as implied correlation smile (or frown) and implied correlation term-structure.

\subsection{Analysis of the Dependence Structure Between two Futures}
\label{ss:AnalysisOfTheDependenceStructureBetweenTwoFutures}

As for distribution functions in our model, the dependence structure between two futures with given maturities is not readily available.
In order to analyze the dependence between futures prices at a future time horizon created by our model we need to work with the joint characteristic function.
As we have seen, it is possible to recover marginal and joint distribution functions directly from the joint characteristic function.
It is also possible to recover via analytical expressions the copula function and the copula density that characterize the dependence between two futures at the chosen time horizon.
These analytical expressions are also valid in the more general context of a two-dimensional stochastic process defined by means of its joint characteristic function.

Let $0 < T \leq T_1,T_2$ with $T$ the chosen time horizon and $T_1$, $T_2$ the expiry dates of the pair of futures.
The joint density $g(.,T,T_1,T_2)$ of $X(T)=\left(X_1(T),X_2(T) \right)$ is recovered by two-dimensional Fourier inversion of the joint characteristic function as
\begin{equation}
\label{JointPDFFormula}
g(x_1,x_2,T,T_1,T_2)=\frac{1}{4\pi^2}\int^{+\infty}_{-\infty}{\int^{+\infty}_{-\infty}{e^{-i(u_1 x_1+u_2 x_2 )}\phi(u_1,u_2,T,T_1,T_2)du_1}du_2}.
\end{equation}

The marginal densities of $X_1(T)$ and $X_2(T)$ are denoted by $g_1(.,T,T_1)$ and $g_2(.,T,T_2)$, respectively, and are recovered as
\begin{align}
g_1(x_1,T,T_1) & = \frac{1}{\pi}\int^{+\infty}_{0}{\text{Re}\left[e^{-iu x_1}\phi(u,0,T,T_1,T_2)\right]du}, \label{MarginalPDF1Formula} \\
g_2(x_2,T,T_2) & = \frac{1}{\pi}\int^{+\infty}_{0}{\text{Re}\left[e^{-iu x_2}\phi(0,u,T,T_1,T_2)\right]du}. \label{MarginalPDF2Formula}
\end{align}

The marginal and joint distribution functions, $G_1,G_2$ and $G$ are also directly recovered from $\phi$ using expressions (\ref{MarginalCDF1Formula}), (\ref{MarginalCDF2Formula}) and  (\ref{JointCDFFormula}).
The dependence structure between the futures prices can now be recovered directly from $\phi$ as the copula function between the log-returns $X_1(T)$ and $X_2(T)$.
This copula is denoted by $C(.,T)$. Note that, expressed as a copula (or a copula density), the dependence structure between the log-returns is the same as the dependence between the prices themselves.
In the remainder of the section and for readability, we drop the explicit reference to $T_1$ and $T_2$ in the expressions.

\begin{proposition}
\label{Prop:CopulaFromCharFunc}
The copula function describing the dependence between $X_1(T)$ and $X_2(T)$ and the corresponding copula density can be recovered, for $(v_1,v_2) \in [0,1]^2$, respectively, as
\begin{align}
C(v_1,v_2,T)
=& \frac{e^{a_1G^{-1}_1(v_1,T)+a_2G^{-1}_2(v_2,T)}}{4\pi^2}
\nonumber
\\
& \int^{+\infty}_{-\infty}{\int^{+\infty}_{-\infty}{ \frac{e^{-i(u_1 G^{-1}_1(v_1,T)+u_2 G^{-1}_2(v_2,T) )}\phi(u_1+ia_1,u_2+ia_2,T)}{(a_1-i u_1)(a_2-i u_2)} du_1}du_2},
\\
c(v_1,v_2,T)
=& \frac{\int^{+\infty}_{-\infty}{\int^{+\infty}_{-\infty}{e^{-i\left(u_1 G^{-1}_1(v_1,T) +u_2 G^{-1}_2(v_2,T) \right)}\phi(u_1,u_2,T)du_1}du_2} }
{\int^{+\infty}_{-\infty}{e^{-iu G^{-1}_1(v_1,T)}\phi(u,0,T)du}
\int^{+\infty}_{-\infty}{e^{-iu G^{-1}_2(v_2,T)}\phi(0,u,T)du}},
\end{align}
where $G^{-1}_1$ and $G^{-1}_2$ are the inverse cumulative distribution functions of $X_1(T)$ and $X_2(T)$.
\end{proposition}

We prove this result in appendix \ref{a1:Proofs}.

The dependence structure created by the model between $X_1(T)$ and $X_2(T)$ is entirely described by the copula function of $C(.,T)$ that is recovered from $\phi$ using expressions in Proposition \ref{Prop:CopulaFromCharFunc}.
This copula function depends on the chosen time horizon $T$ and we actually have a term-structure of dependence that can be obtained from $\phi$.
The indexing by $T$ of the copula $C$ should be understood as a time-horizon, since it describes, seen from $t=0$, the distribution of the random vector $\left(G_1(X_1(T),T),G_2(X_2(T),T) \right)$.
The chosen model produces a term-structure of dependence (i.e. a term-structure of copulas) which should not be confused with a time dependent copula.

Figure \ref{Fig:copula} plots the copula function and copula density between two futures prices obtained with the SV2F model and parameters as given in Table \ref{tab:ModelParameters}.
The chosen futures respective maturities are $T_1=0.25$ years and $T_2=0.75$ years.
The chosen time horizon is $T=0.25$ years.
The obtained copula and its density appear to be smooth.
A different choice of time horizon (e.g. a shorter one) would have led to a different dependence between the two futures.

\begin{figure}[ht]
\centering
	\includegraphics[height=6.0cm]{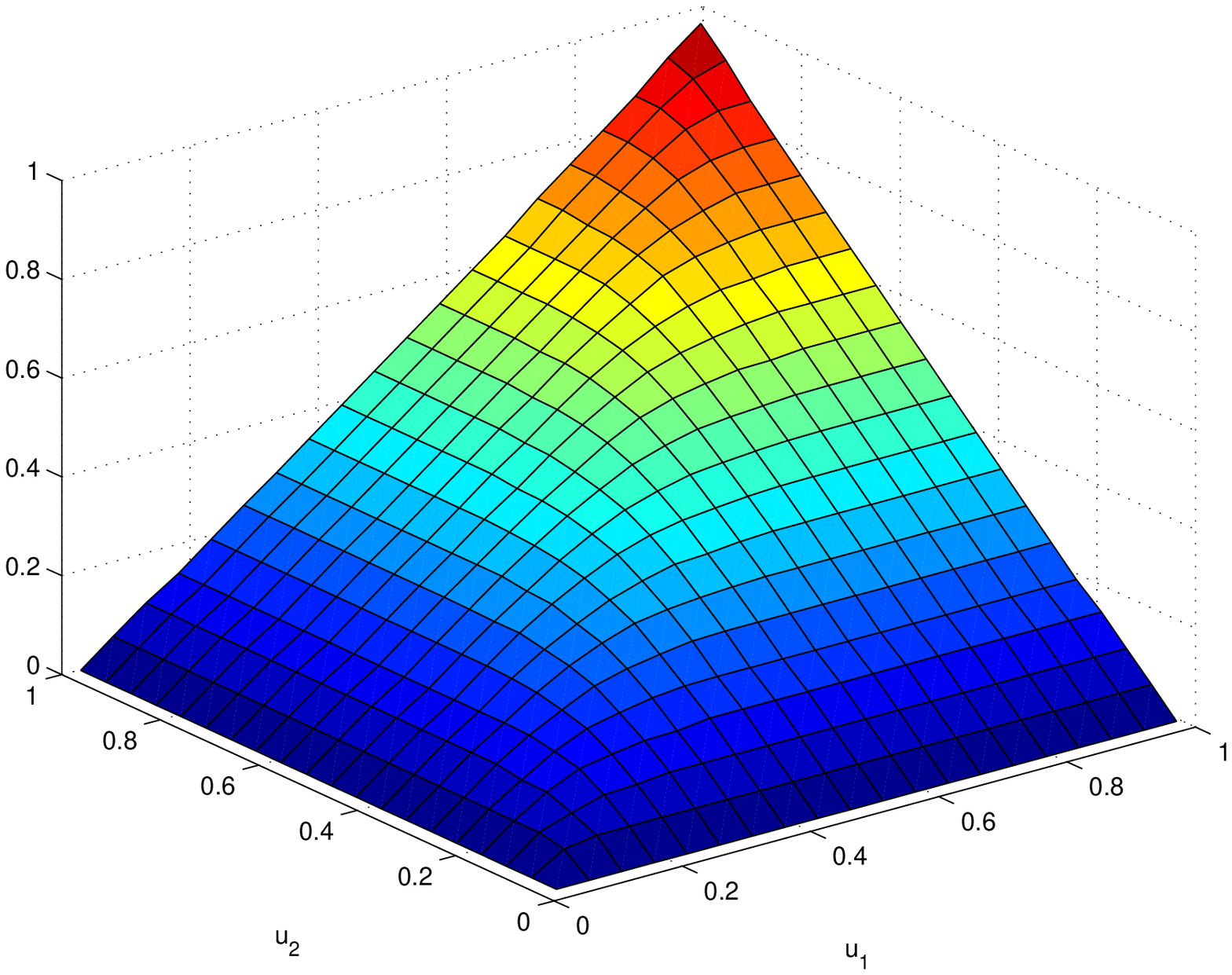} \quad
	\includegraphics[height=6.0cm]{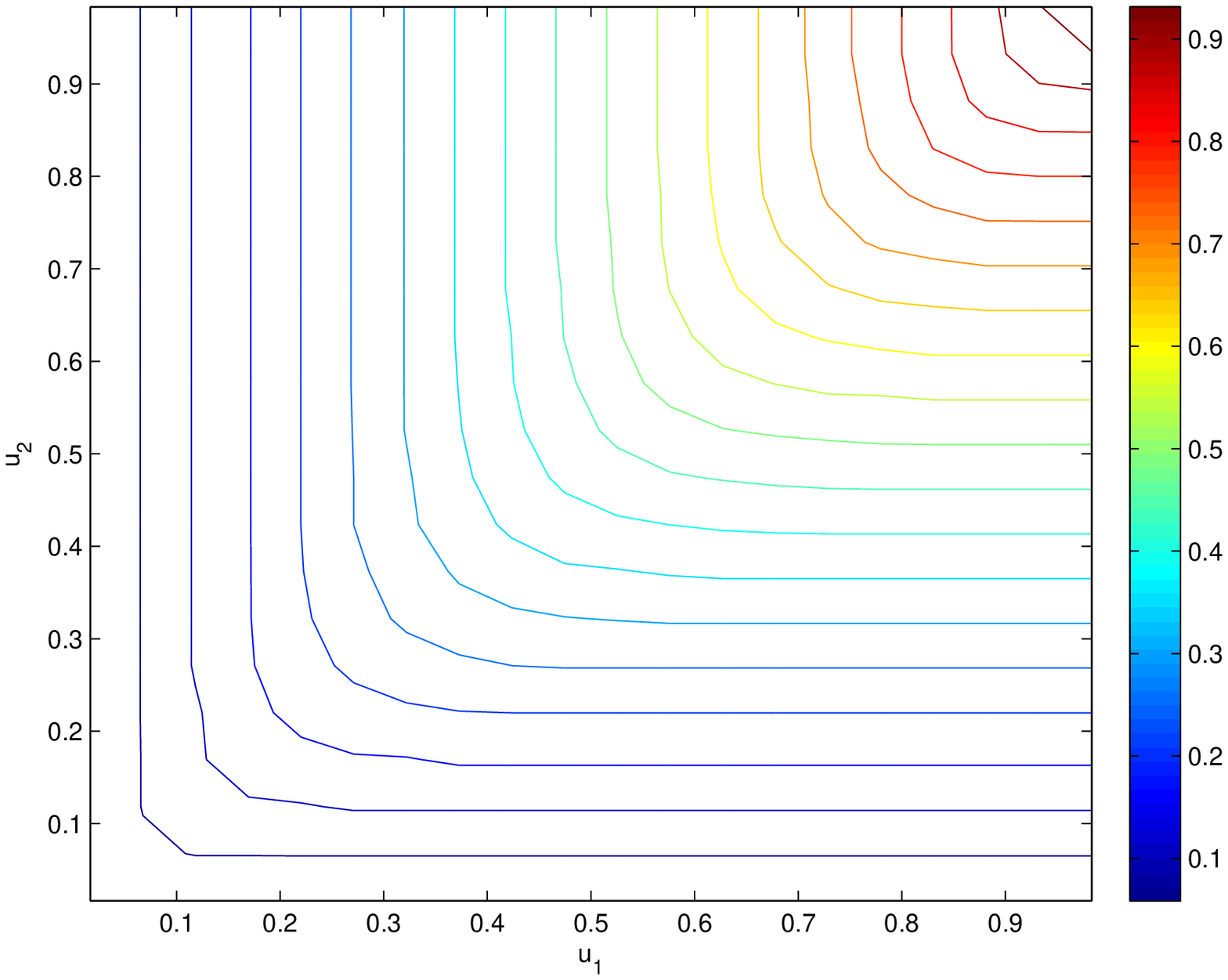}
	
	\includegraphics[height=6.0cm]{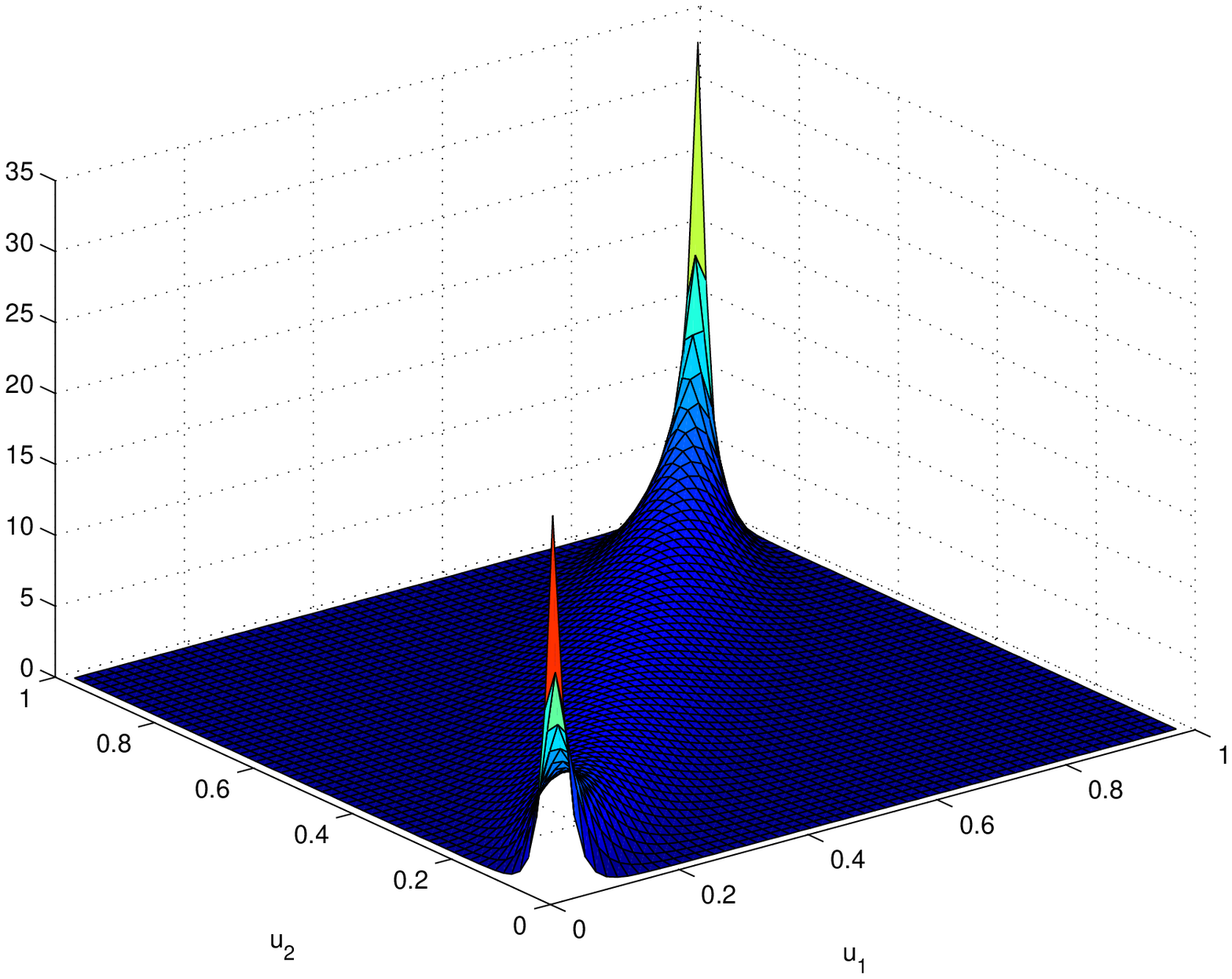} \quad
	\includegraphics[height=6.0cm]{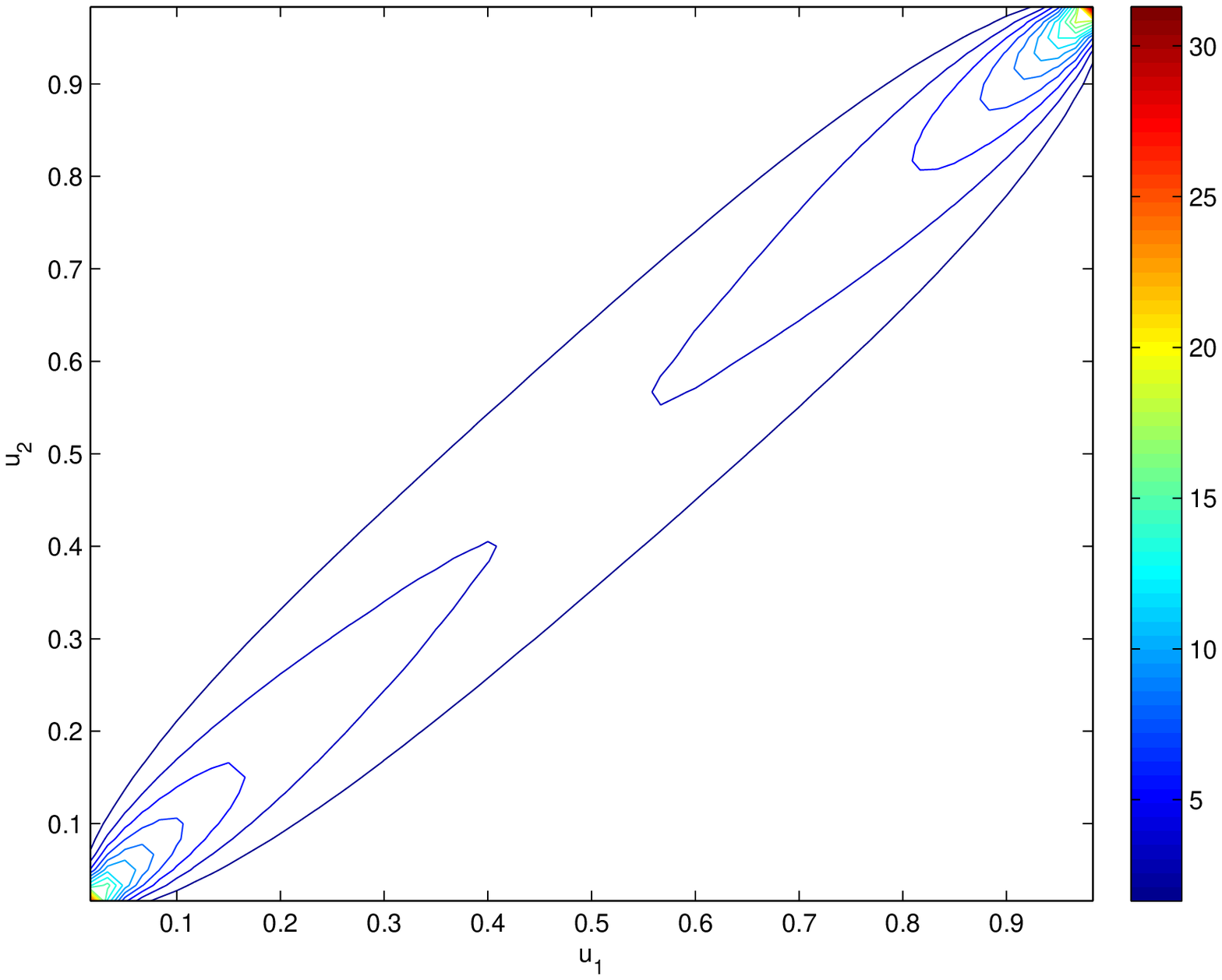}
		
	\caption{
    \label{Fig:copula}
Copula function and copula density representing the dependence structure between $F(T,T_1)$ and $F(T,T_2)$, for $T=T_1=0.25$ years and $T_2=0.75$ years,
obtained with SV2F model and parameters as given in Table \ref{tab:ModelParameters}.}
\end{figure}

To assess the dependence between $X_1(T)$ and $X_2(T)$ with a single number instead of a function one can rely on concordance and dependence measures.
Two usual concordance measures are Kendall's tau and Spearman's rho.
For $(X_1(T),X_2(T))$ these measures are denoted by $\tau_K(X_1,X_2,T)$ and $\varrho_S(X_1,X_2,T)$, respectively.
Two usual dependence measures are Schweizer-Wolf's sigma and Hoeffding's phi.
For $(X_1(T),X_2(T))$ these measures are denoted by $\sigma_{SW}(X_1,X_2,T)$ and $\Phi_H(X_1,X_2,T)$, respectively.
Concordance and dependence measures can be expressed as functions of the copula of $(X_1(T),X_2(T))$.
The four measures mentioned here can be written as
\begin{align}
\tau_K(X_1,X_2,T) & = 4 \iint_{[0,1]^2}{C(v_1,v_2,T)c(v_1,v_2,T)dv_1dv_2} -1, \\
\varrho_S(X_1,X_2,T) & = 12 \iint_{[0,1]^2}{C(v_1,v_2,T)dv_1dv_2} -3, \\ 
\sigma_{SW}(X_1,X_2,T) & = 12 \iint_{[0,1]^2}{\mid C(v_1,v_2,T) - C^{\bot}(v_1,v_2) \mid dv_1 dv_2}, \\
\Phi_H(X_1,X_2,T) & = 90 \iint_{[0,1]^2}{\mid C(v_1,v_2,T) - C^{\bot}(v_1,v_2) \mid^2 dv_1 dv_2},
\end{align}
where $C^{\bot}(v_1,v_2)=v_1 v_2$ is the product copula associated with independence between variables.
We refer to \citet{Nelsen2006} for more details and properties of these measures.

\section{Calibration to Market Data and Empirical Considerations}
\label{s:Calibration}

In this section we consider empirical data for WTI.
We calibrate the two factor version of our model on different dates corresponding to different market situations.
We then analyze the term-structure of dependence produced by the calibrated model and the implied correlations obtained when pricing calendar spread options.

\subsection{Data}
\label{ss:Data}

For empirical applications, we use three sets of WTI market data.
Each dataset corresponds to a cross-section of futures and options closing prices taken on a given date.
We have chosen three dates as representatives of different market situations.
The first date is December 10, 2008 and can be considered as taken during the financial crisis period, as it was just months after the default of Lehman Brothers.
On that date, implied volatilities of WTI vanilla options were above 80\% for the shortest maturities and the OVX index was calculated above 90\%.
The second date is March 9, 2011 and corresponds to a market that is recovering from the deepest states of the crisis.
The third considered date is April 9, 2014 that can be seen as a ``\textit{back to normal}'' market situation, at least from the standpoint of market prices.

Futures and options on WTI are traded and quoted on the NYMEX.
Interest rates data and closing prices for futures as well as vanilla and calendar spread options were obtained from Bloomberg and Datastream. 

\subsection{Calibration to Vanilla Options}
\label{ss:CalibrationToVanillaOptions}

Models considered in this paper, namely SV2F (two-factor version of the proposed stochastic volatility model) and CS2F (two-factor version of the Clewlow-Strickland model),
can be fitted to a cross-section of observed vanilla options prices.
For each dataset we calibrate these models by minimizing the sum of squared errors between model and observed prices.
This is equivalent to a calibration to implied volatilities.
For a given dataset, the calibrated model parameter set $\theta^*$ is obtained as
\begin{equation}
\theta^* = \text{arg}\min_{\theta \in \Theta}\sum^{N_T}_{i=1}{\sum^{N_K}_{j=1}{\left(O(K_j,T_i,T_i;\theta) - O^{Obs}(K_j,T_i,T_i) \right)^2}},
\end{equation}
where $\Theta$ is the set of feasible model parameters, $N_T$ the number of maturities in the options set,
$N_K$ the number of strikes for each maturity (without loss of generality we consider the same number of strikes to be available for each maturity).
$O(.;\theta)$ denotes the option price obtained using the chosen model with parameter $\theta$ and $O^{Obs}(.)$ denotes the corresponding observed price.
In the considered datasets we work with five maturities, ranging from two months to four years (hence $N_T=5$), and seven strikes for each maturity, that are specified
in terms of moneyness with respect to the corresponding futures price.
Specifically, these strikes are 60\%, 80\%, 90\%, 100\%, 110\%, 120\% and 150\% (hence $N_K=7$).

Once the minimization programs have been solved, it is possible to measure the quality of the obtained calibration.
The fit quality can be measured as mean absolute error or root mean squared error on prices.
These errors are calculated as
\begin{align*}
MAE & = \sum^{N_T}_{i=1}{\sum^{N_K}_{j=1}{\frac{\left|O(K_j,T_i;\theta^*) - O^{Obs}(K_j,T_i) \right|}{N_K N_T}}},  \\
RMSE & = \sqrt{\sum^{N_T}_{i=1}{\sum^{N_K}_{j=1}{\frac{\left(O(K_j,T_i;\theta^*) - O^{Obs}(K_j,T_i) \right)^2}{N_K N_T}}}}.
\end{align*}

The same measures of fit can be applied to implied volatilities instead of prices.
Table \ref{tab:calibration_fit} presents, for each dataset, these error measures obtained with the calibrated models.
Due to the presence of implied volatility smiles along the strike-axis, the CS2F model is not able to closely match the observed prices.
In contrast, the SV2F model can provide a proper fit to both the strike-structure and the term-structure of implied volatilities.
The error measures obtained with SV2F model appear to be around half a volatility point for MAE and around three quarters of a point for RMSE,
which can be regarded as very good given the large strike and maturity spans of the options sets.

\begin{table}[ht]
  \centering
	\begin{tabular}{ccccccc}
	\toprule
		\phantom{ab} & \multicolumn{2}{c}{MAE} & \multicolumn{2}{c}{MAE (ATM)} & \multicolumn{2}{c}{RMSE} \\
		\cmidrule(lr){2-3} \cmidrule(lr){4-5} \cmidrule(lr){6-7}
		Date & Price & Vol. & Price & Vol. & Price & Vol. \\
		\midrule \\
		\multicolumn{4}{l}{\textit{CS2F Model}} \\
		Dec. 2008 &	$0.4607$ & $0.0156$ & $0.5205$ & $0.0148$ & $0.5827$ & $0.0180$ \\
		Mar. 2011 & $0.5394$ & $0.0169$ & $0.3941$ & $0.0082$ & $0.6803$ & $0.0243$ \\
		Apr. 2014 &	$0.2884$ & $0.0162$ & $0.2644$ & $0.0074$ & $0.3616$ & $0.0226$ \\ \\
		\multicolumn{4}{l}{\textit{SV2F Model}} \\
		Dec. 2008 &	$0.1278$ & $0.0052$ & $0.1627$ & $0.0062$ & $0.1777$ & $0.0066$ \\
		Mar. 2011 & $0.2071$ & $0.0056$ & $0.2419$ & $0.0056$ & $0.2922$ & $0.0078$ \\
		Apr. 2014 & $0.1113$ & $0.0043$ & $0.0965$ & $0.0026$ & $0.1517$ & $0.0059$ \\ \\
		\bottomrule
	\end{tabular}
\caption{MAE and RMSE error measures, of prices and implied volatilities, for SV2F and CS2F models calibrated to vanilla option prices.
\textit{Left panel} is for MAE on the whole matrix of options, \textit{central panel} is for MAE on at-the-money options and \textit{right panel} is for RMSE on the whole matrix of options.}
\label{tab:calibration_fit}
\end{table}


\subsection{Results}
\label{ss:Results}

Figure \ref{Fig:clibratedimpliedvol} plots, for each dataset, implied volatility corresponding to the market data and implied volatility obtained with the SV2F model calibrated
to vanilla option prices.
Plots in the left column give evidence of the implied volatility smile in our datasets.
It can also be noted that the convexity and skew (at-the-money slope) of these smiles vary with maturity.
This maturity effect is particularly present in March 2011 data. Plots in the right column represent at-the-money volatility term-structure.
They provide evidence for the empirical Samuelson volatility effect in the market prices of at-the-money vanilla options.
Figure \ref{Fig:clibratedimpliedvol} shows that our model, once calibrated to vanilla option prices,
is able to properly reproduce the empirical stylized facts for implied volatility, namely presence and maturity-dependency of the smile and Samuelson effect.

Figure \ref{Fig:dependencemeasures1} plots, for the datasets of December 2008 and March 2011,
measures of concordance and dependence between two futures prices produced by the calibrated SV2F model.
This figure represents a type of dependence term-structure that is of interest for WTI futures market participants.
It corresponds to the case where the time horizon and the first futures expiry are both held constant, while the difference between futures expiries varies.
We observe that, as the difference between expiries increases, the pair of futures becomes less dependent which is line with the intuition one can have \textit{a priori}.
%
We call this phenomenon the \textit{Samuelson correlation effect}.
This phenomenon is a desirable feature for a model to be used by a price maker quoting and trading a range of products written on WTI.
The presented empirical applications show it is properly reproduced by the proposed model.


Figure \ref{Fig:dependencemeasures2} plots, for the datasets of December 2008 and March 2011,
measures of concordance and dependence between two futures prices produced by the calibrated SV2F model.
It corresponds to the case where the time horizon varies while the difference between futures expiries is held constant.
For the December 2008 case, the level of dependence between the futures is little affected by the time-horizon.
For the March 2011 case, as the time-horizon increases, the pairs of futures with $6$-month difference between their maturities become more dependent.
Here the intuition does not lead to a particular structure that is desirable, namely increasing or decreasing as the time horizon varies.

Figure \ref{Fig:clibratedimpliedcorrel} plots the implied correlation strike and maturity structures from spread option prices
obtained with the SV2F model calibrated to the datasets of December 2008 and March 2011.
The considered spread options have fixed maturity and first futures expiry, while the difference between the two underlying futures expiries varies.
For each pair of underlying futures, the five strikes correspond to a set of shifts applied to the at-the-money spread $F(0,T_1)-F(0,T_2)$.
These shifts are $-10,-5,-2.5,0,+2.5,+5$ and $+10$.
We observe that the obtained term-structures are decreasing.
This observation is in line with the intuition and with what is observed in Figure \ref{Fig:dependencemeasures1}, i.e. with the \textit{Samuelson correlation effect}.
For the March 2011 case, the model produces a non-constant strike structure of implied correlation.
We observe that for larger strikes (out of the money call spread options) implied correlations are lower than at the money.
Lower implied correlations in turn correspond to higher option prices.
For the December 2008 case, the model produces a rather flat strike structure of implied correlation.
Hence, in this case, the prices produced by the model are close to prices that could have been produced using Gaussian copulas for the dependence between futures prices.

Figure \ref{Fig:clibratedimpliedcorrel2} plots the implied correlation strike and maturity structures from spread option prices
obtained with the SV2F model calibrated to the datasets of December 2008 and March 2011.
The considered spread options have maturity and first futures expiry that vary, while the difference between the two underlying futures expiries is held constant.
For each pair of underlying futures, the five strikes correspond to a set of shifts applied to the at-the-money spread $F(0,T_1)-F(0,T_2)$.
These shifts are the same as for Figure \ref{Fig:clibratedimpliedcorrel}.
We observe that the obtained term-structures are decreasing.
We observe that the obtained term-structure for December 2008 is rather flat which is consistent
with the term-structure of concordance and dependence presented in Figure \ref{Fig:dependencemeasures2}.
For the March 2011 case, the implied correlation term-structure is increasing
which is again consistent with concordance and dependence measures in Figure \ref{Fig:dependencemeasures2}.
For the March 2011 case, the model produces a non-constant strike structure of implied correlation.
For the December 2008 case, the model produces a rather flat strike structure of implied correlation.
These strike structures are similar to those found in Figure \ref{Fig:clibratedimpliedcorrel} and the same comments apply.


\begin{figure}[ht]
\centering
	\includegraphics[height=6.0cm]{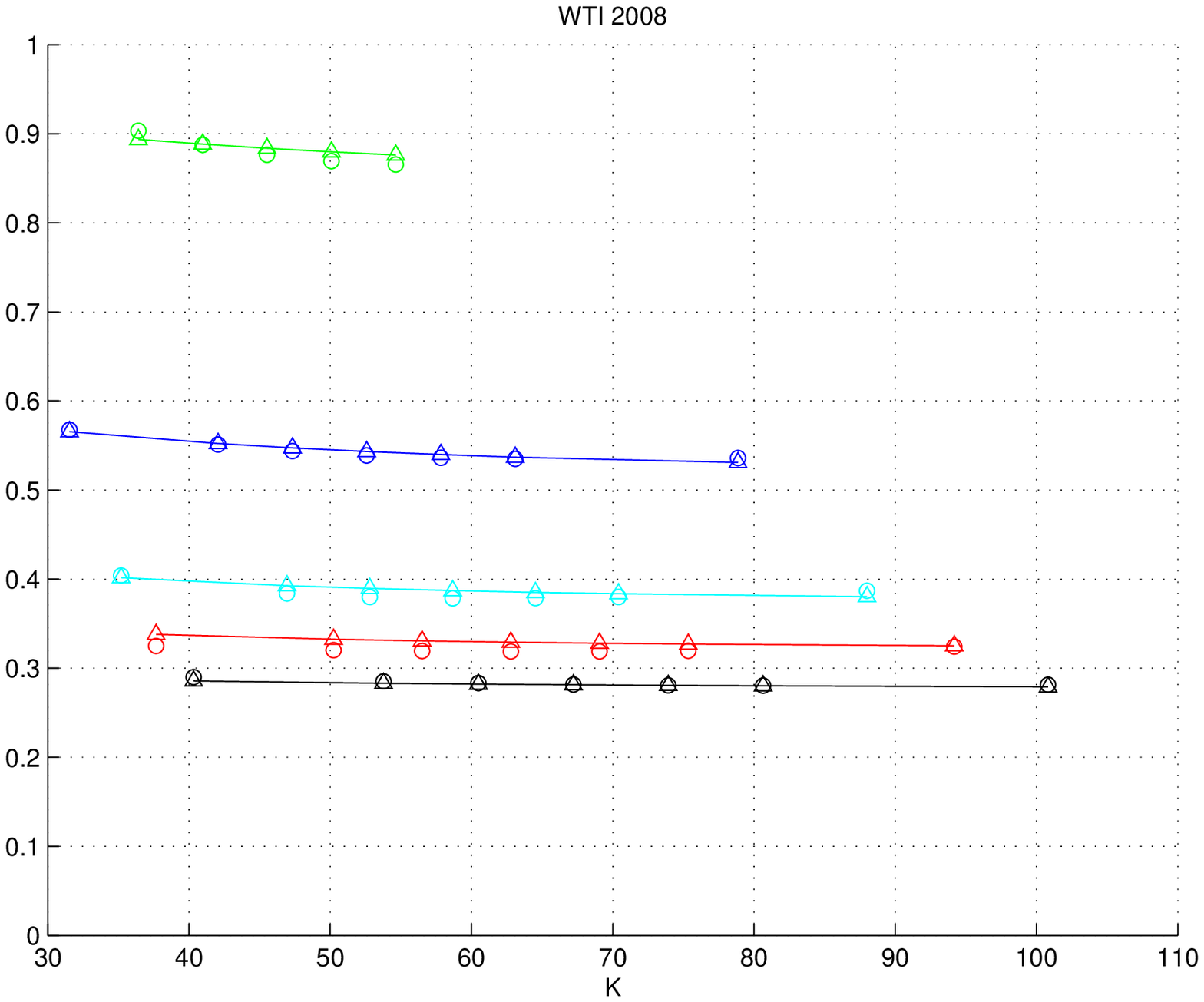} \quad
	\includegraphics[height=6.0cm]{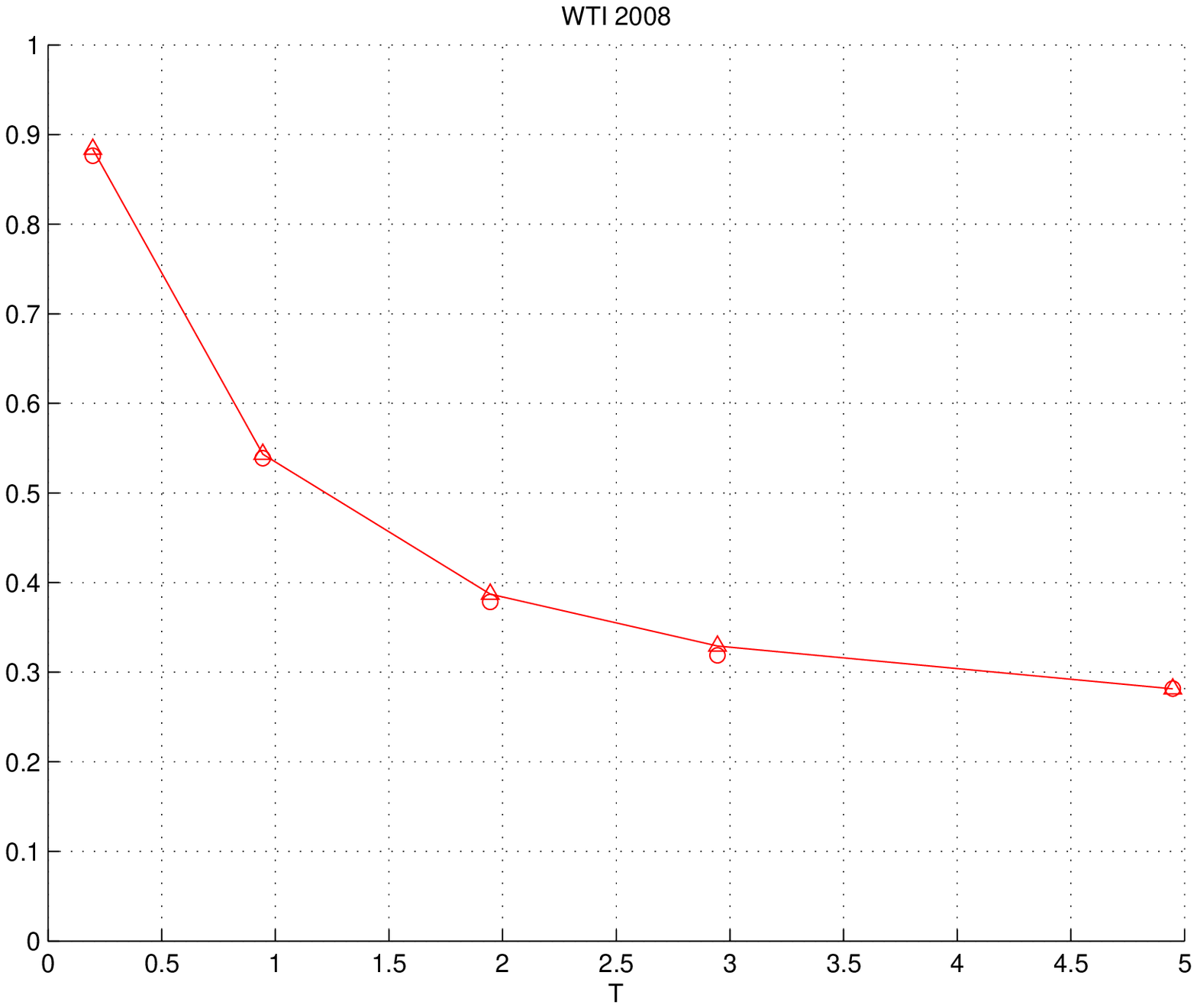}
	
	\includegraphics[height=6.0cm]{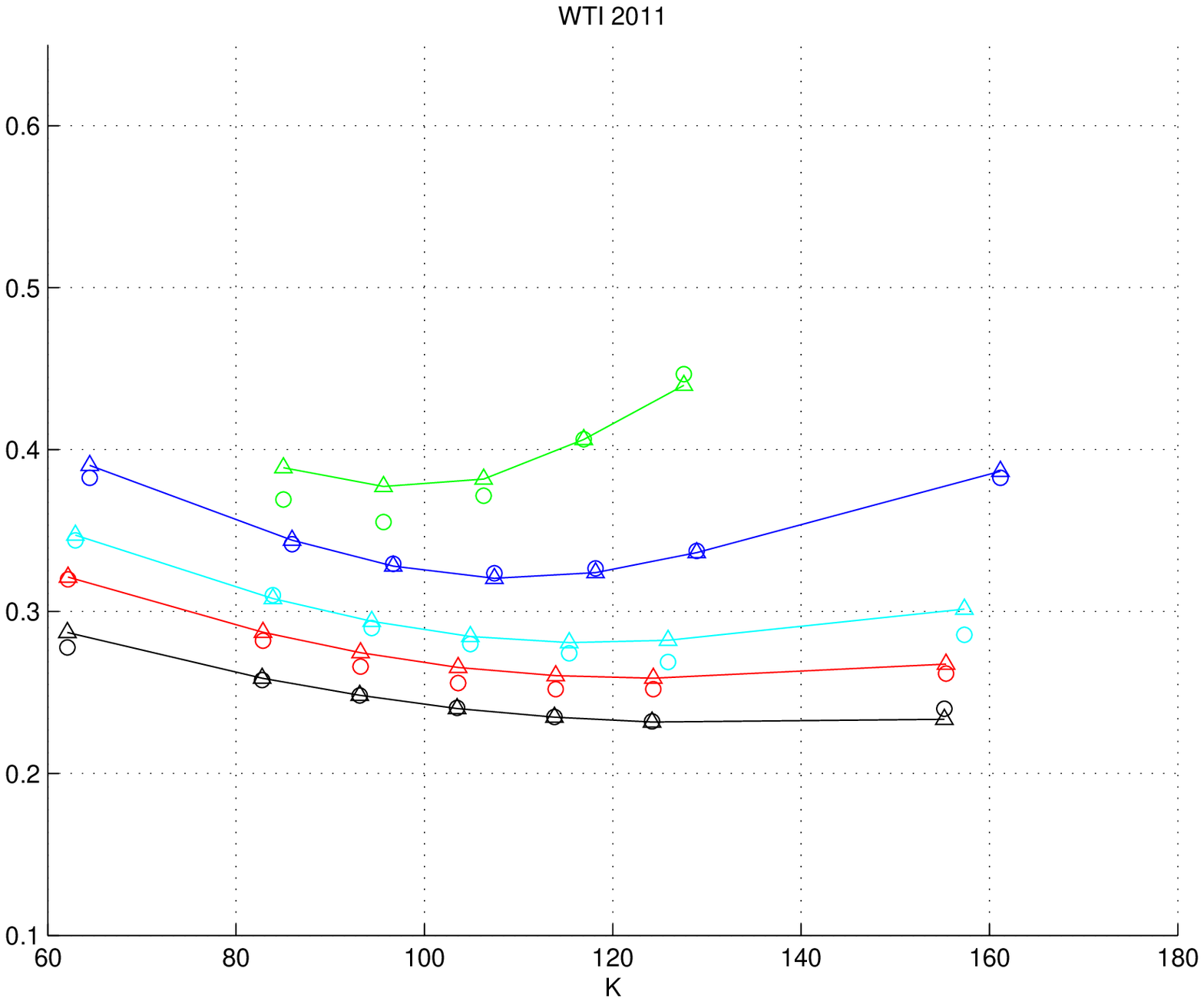} \quad
	\includegraphics[height=6.0cm]{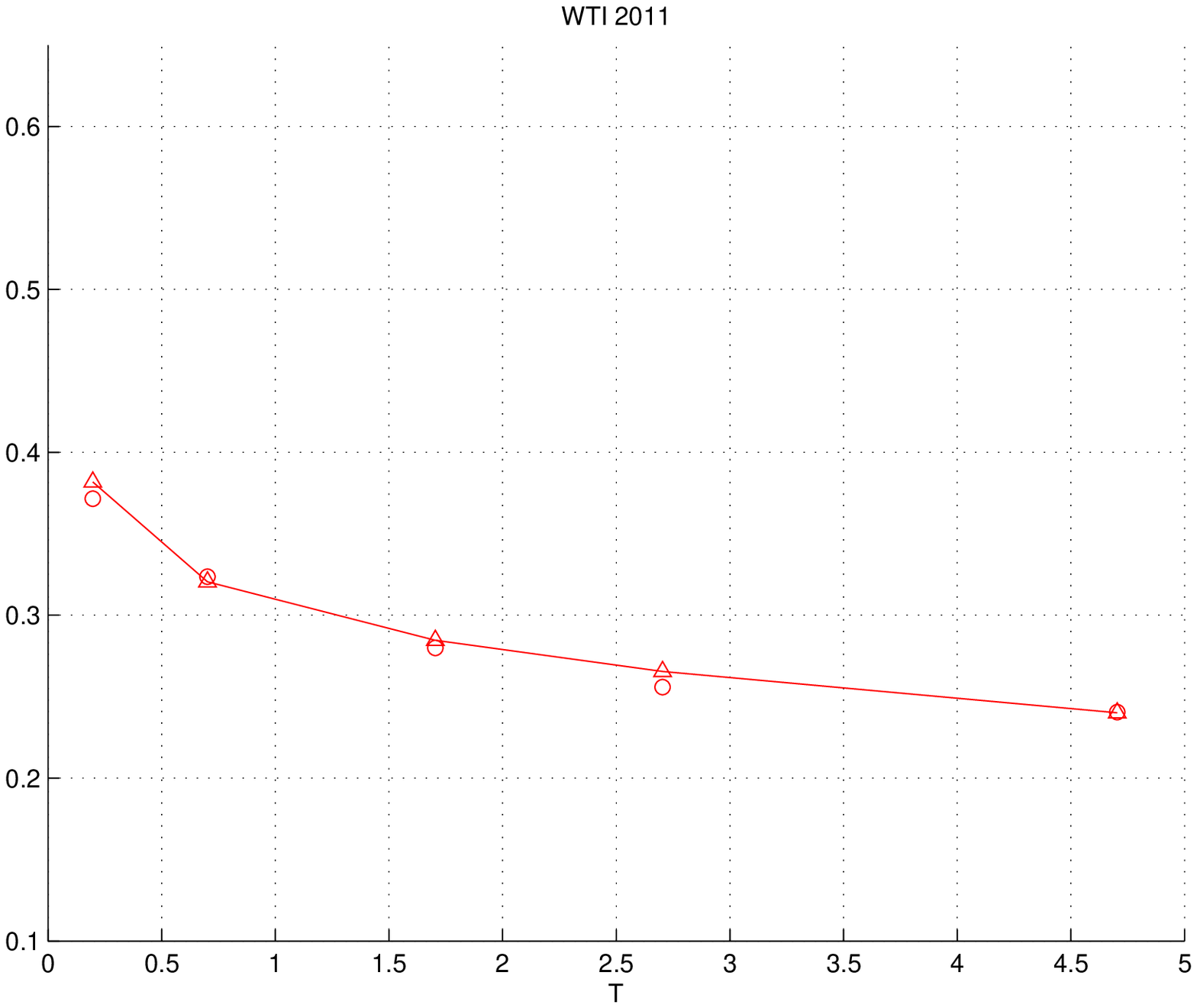}
	
	\includegraphics[height=6.0cm]{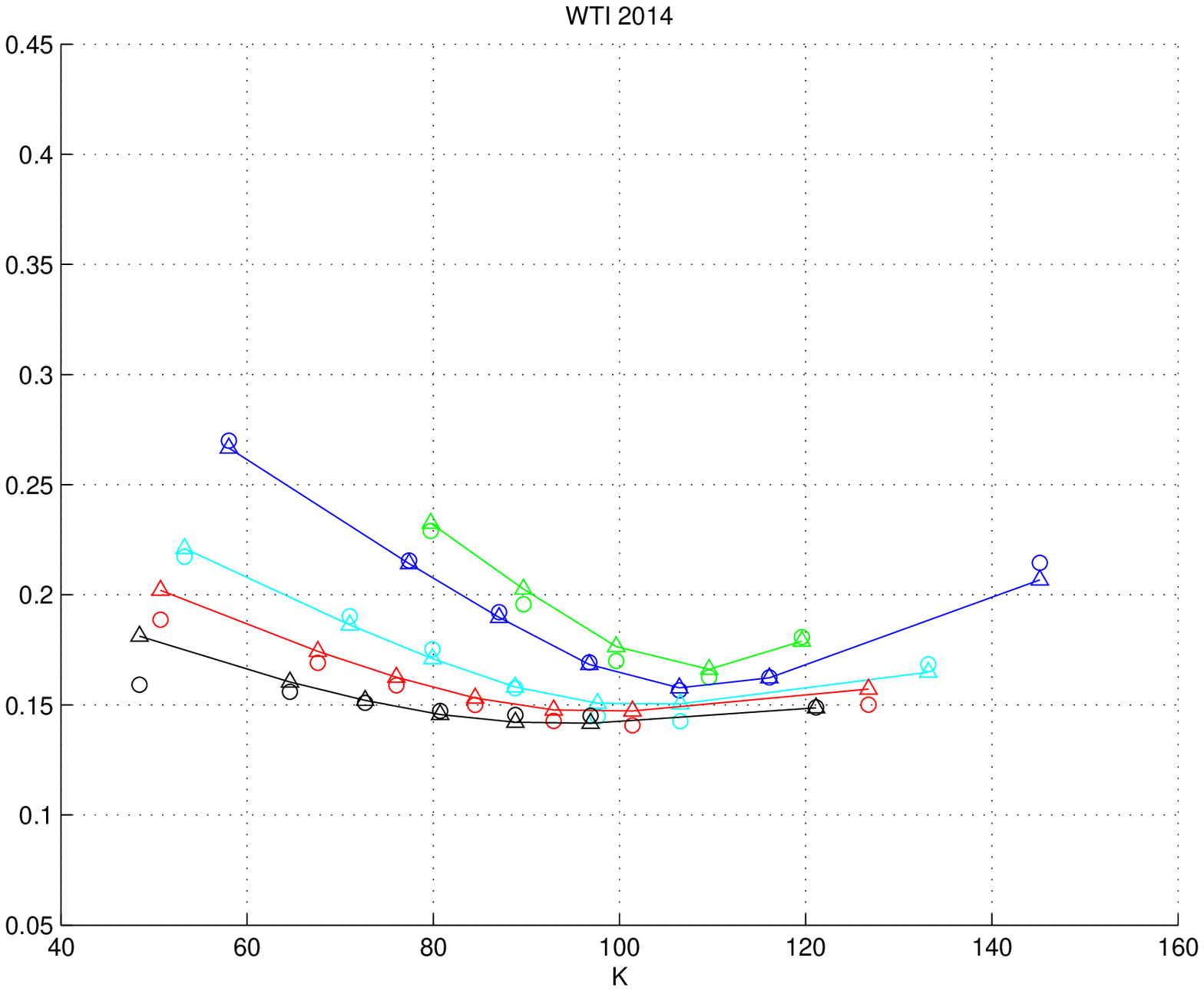} \quad
	\includegraphics[height=6.0cm]{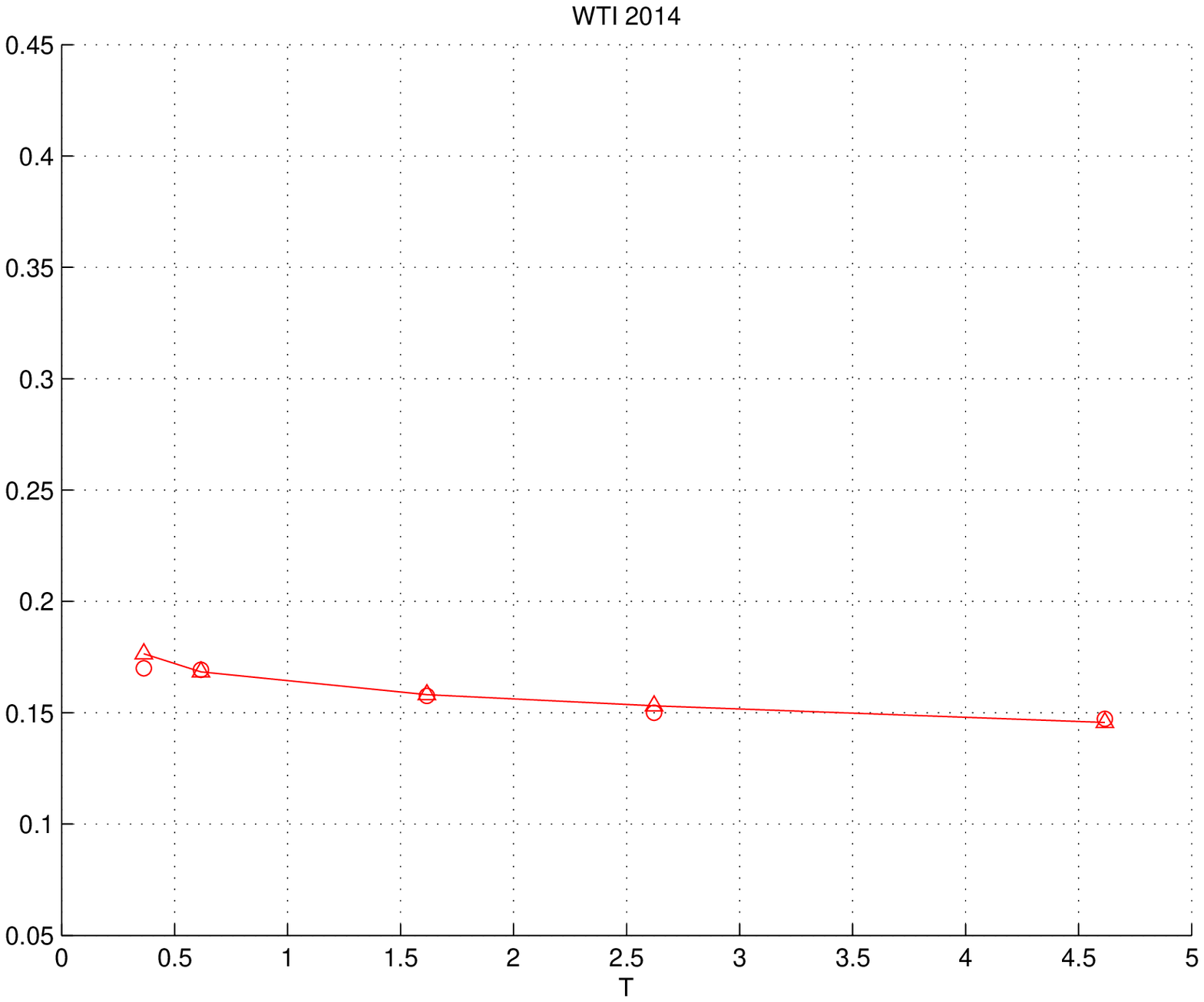}
	
	\caption{\label{Fig:clibratedimpliedvol}
Implied volatilities corresponding to market quotes and obtained with SV2F model calibrated to vanilla options.
\textit{Left column}: implied volatility smiles.
\textit{Right column}: at-the-money implied volatility term-structure.}
\end{figure}



\begin{figure}[ht]
\centering
	\includegraphics[height=7.0cm]{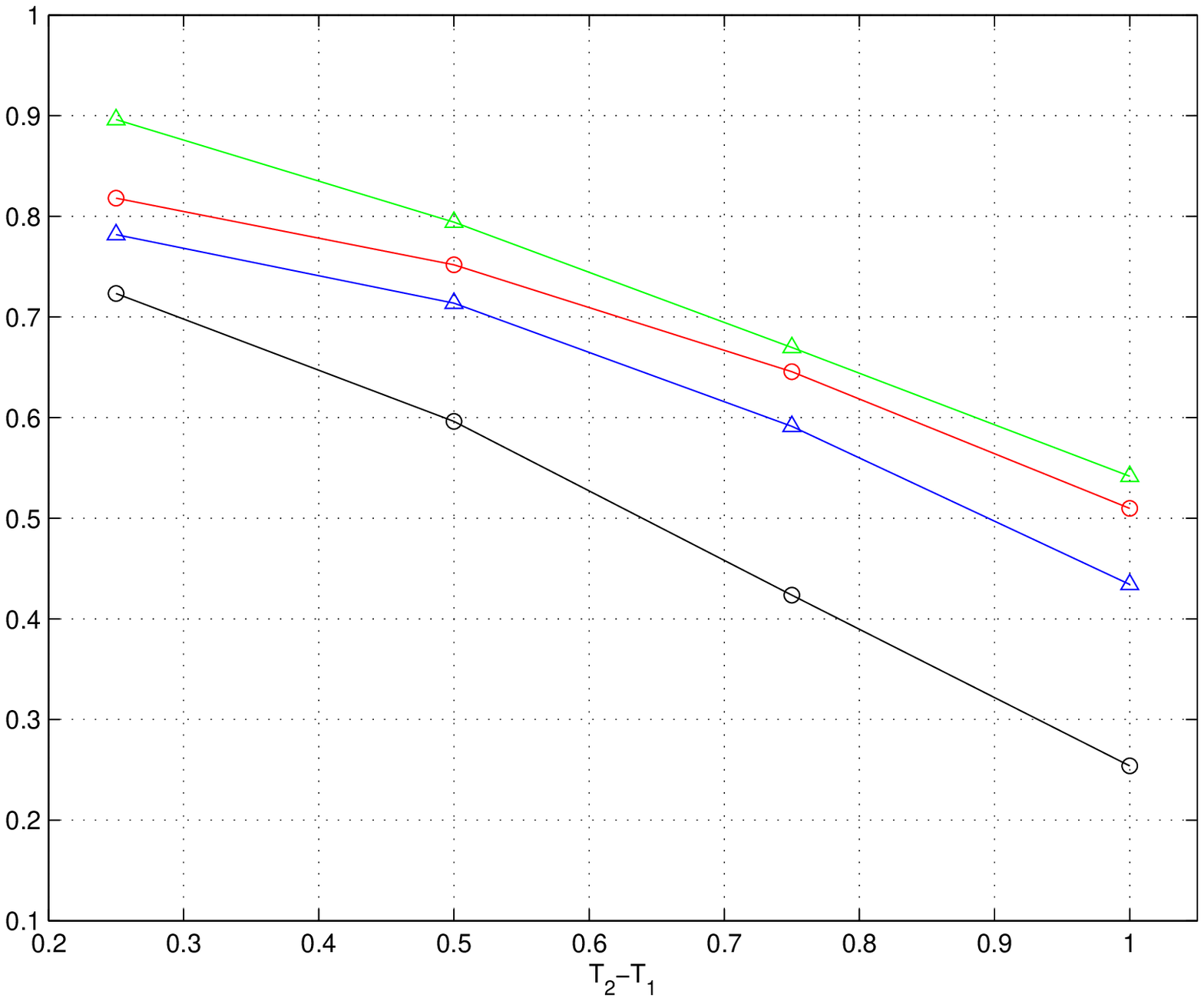}
	
	\includegraphics[height=7.0cm]{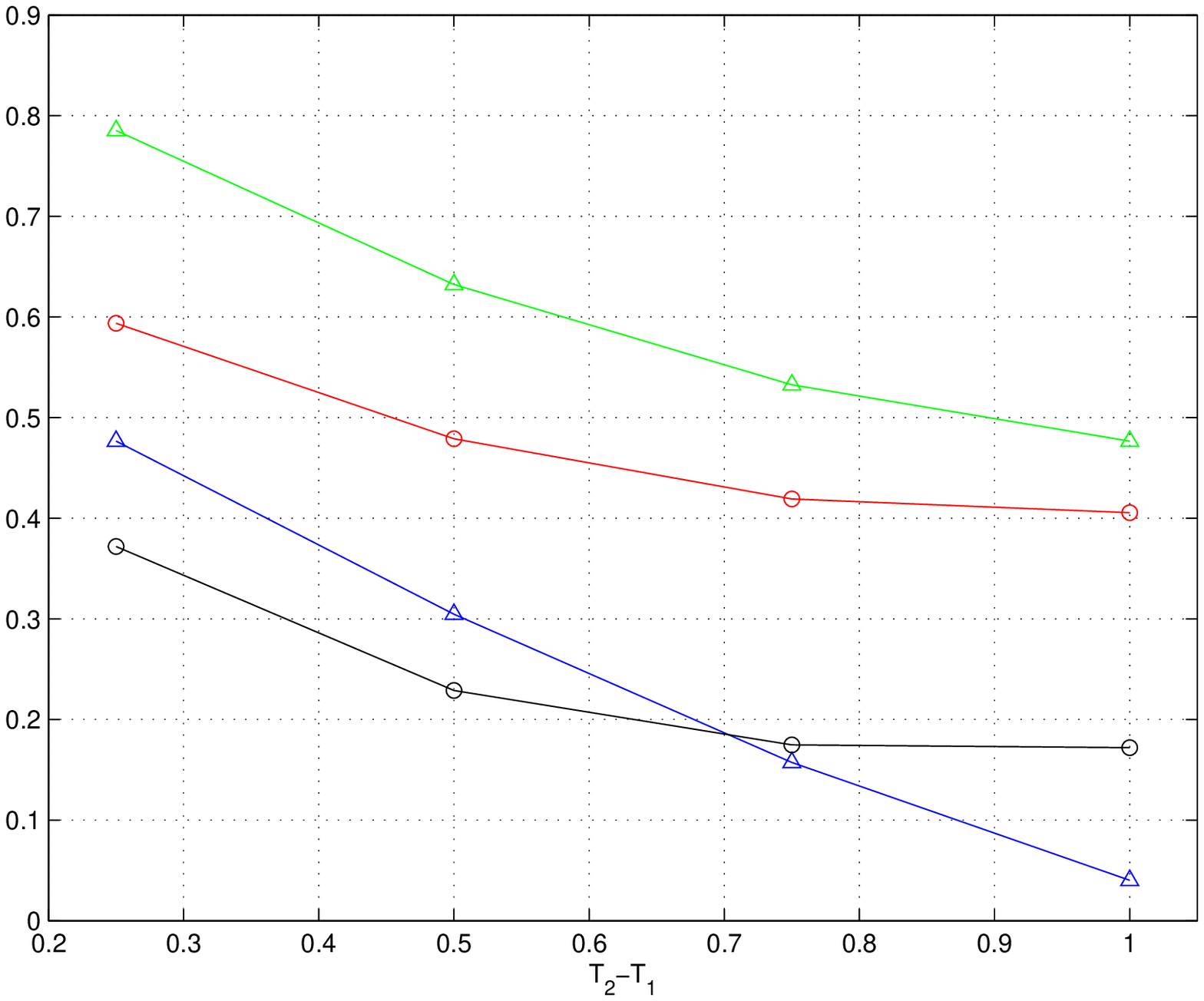}
		
	\caption{\label{Fig:dependencemeasures1}
Term-structure of concordance and dependence measures between $F(T,T_1)$ and $F(T,T_2)$ produced by the SV2F model calibrated market data.
\textit{Upper panel} corresponds to December 2008 data and \textit{lower panel} corresponds to March 2011 data.
Time horizon $T$ and first futures expiry $T_1$ are fixed at $3$ months, $T_2 - T_1$ ranges from $3$ months to $1$ year.
Dependence measures are $\tau_K$ and $\varrho_S$ (respectively, \textit{red} and \textit{green} lines).
Concordance measures are $\sigma_{SW}$ and $\Phi_H$ (respectively, \textit{red} and \textit{black} lines).}
\end{figure}

\begin{figure}[ht]
\centering
	\includegraphics[height=7.0cm]{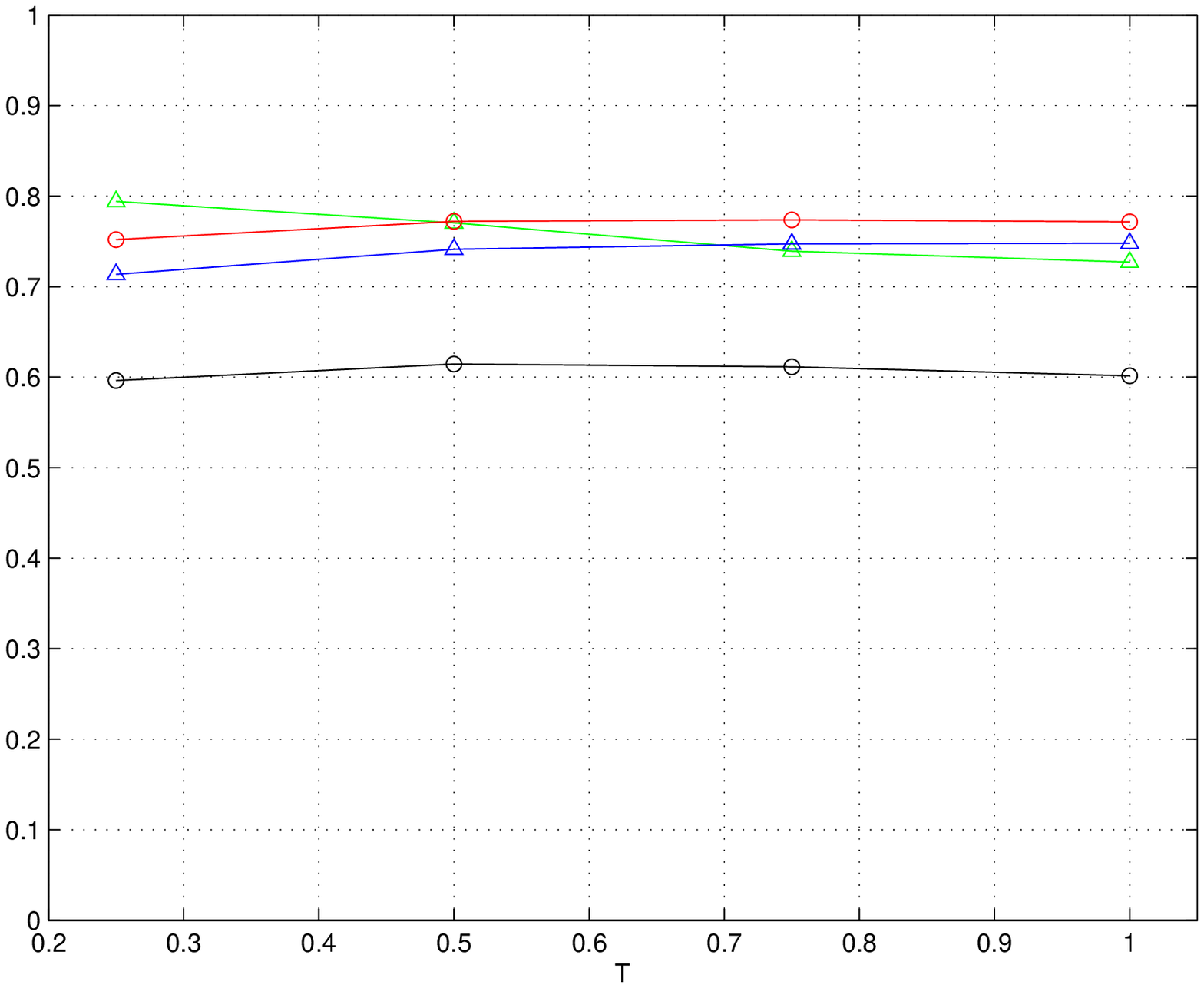}
	
	\includegraphics[height=7.0cm]{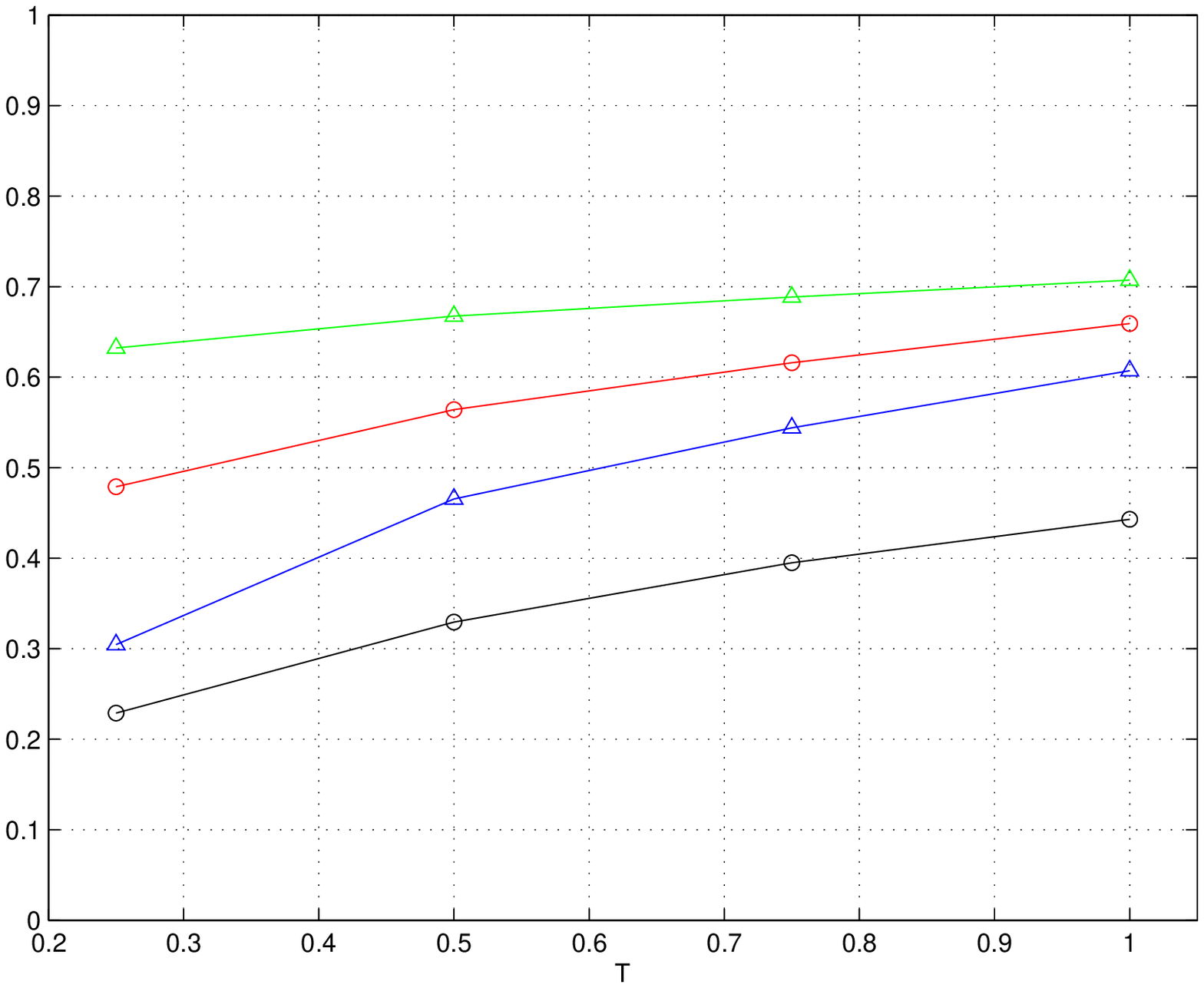}
		
	\caption{\label{Fig:dependencemeasures2}
Term-structure of concordance and dependence measures between $F(T,T_1)$ and $F(T,T_2)$ produced by the SV2F model calibrated market data.
\textit{Upper panel} corresponds to December 2008 data and \textit{lower panel} corresponds to March 2011 data.
Time horizon $T$ and first futures expiry $T_1$ range from $3$ months to $1$ year.
$T_2 - T_1$ is held fixed at $6$ months.
Dependence measures are $\tau_K$ and $\varrho_S$ (respectively, \textit{red} and \textit{green} lines).
Concordance measures are $\sigma_{SW}$ and $\Phi_H$ (respectively, \textit{red} and \textit{black} lines).}
\end{figure}



\begin{figure}[ht]
\centering
	\includegraphics[height=6.0cm]{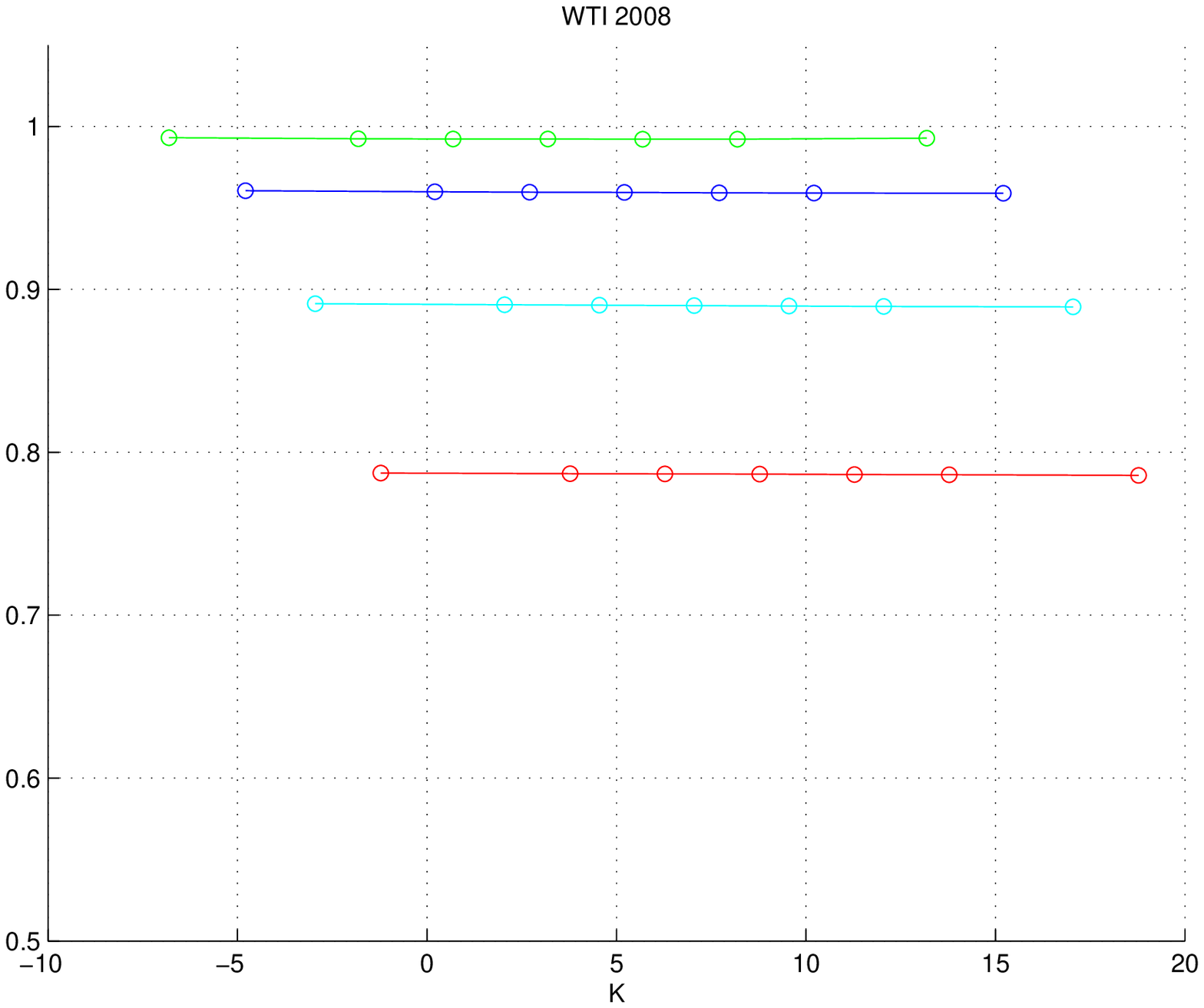} \quad
	\includegraphics[height=6.0cm]{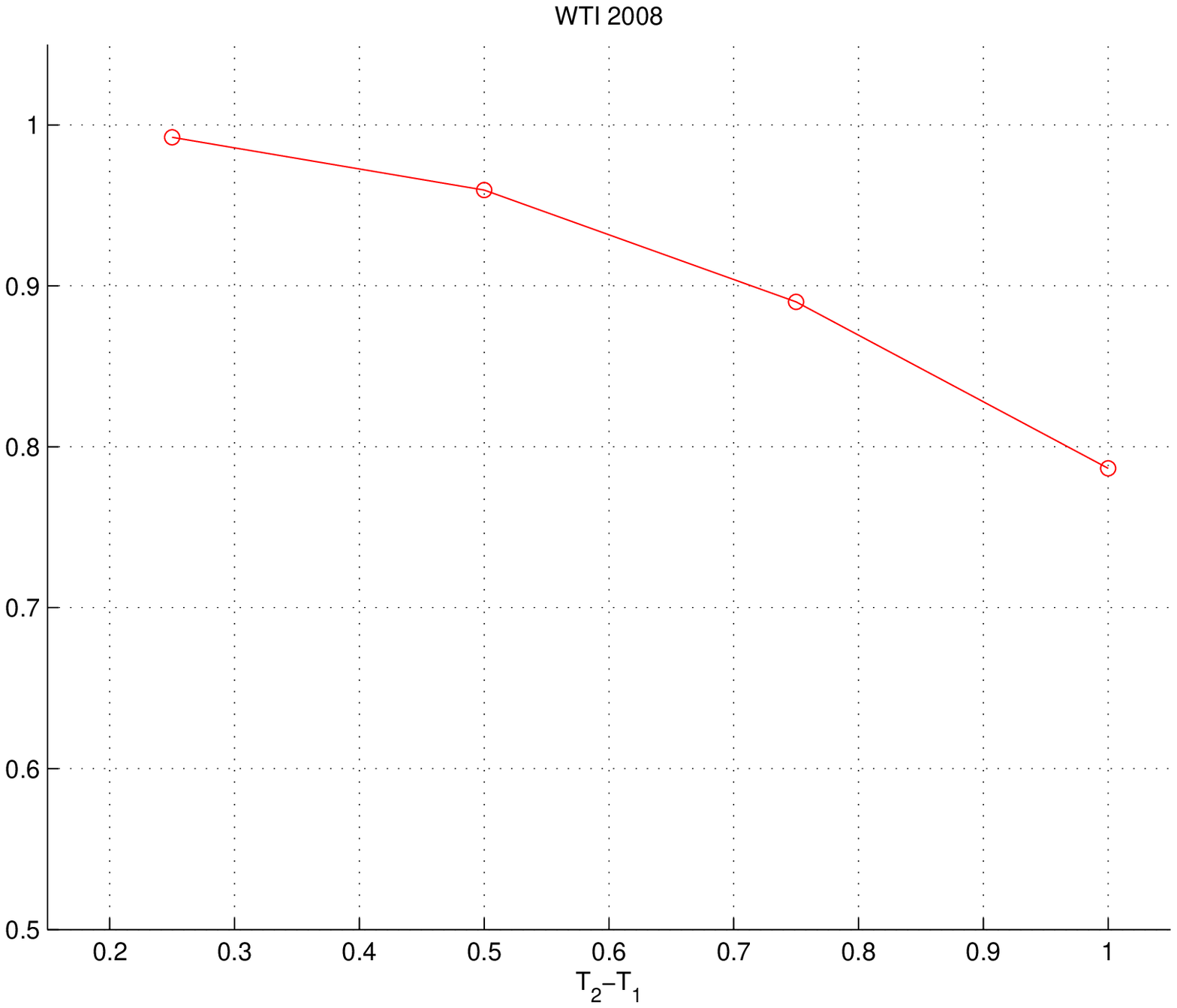}
	
	 \includegraphics[height=6.0cm]{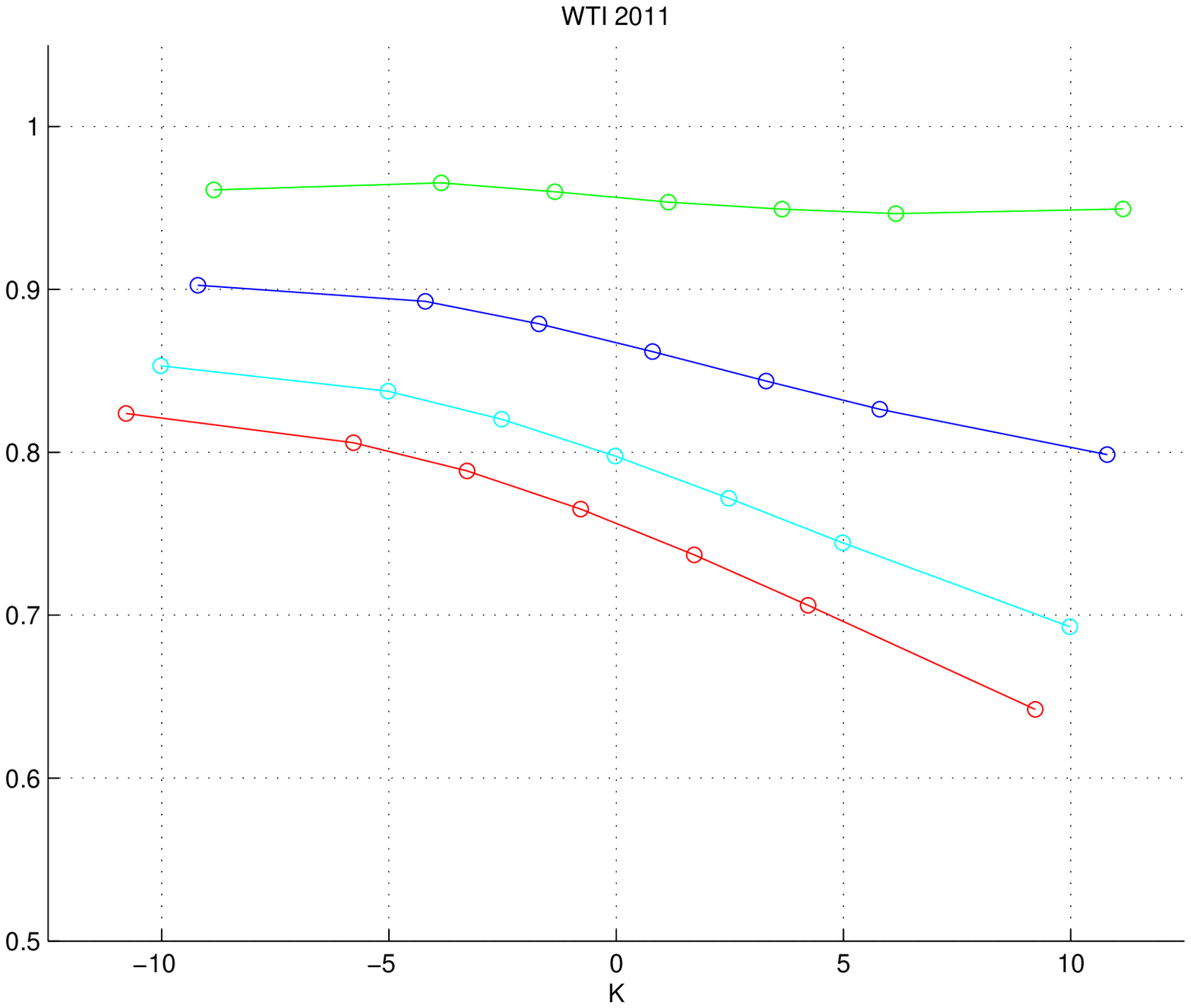} \quad
	 \includegraphics[height=6.0cm]{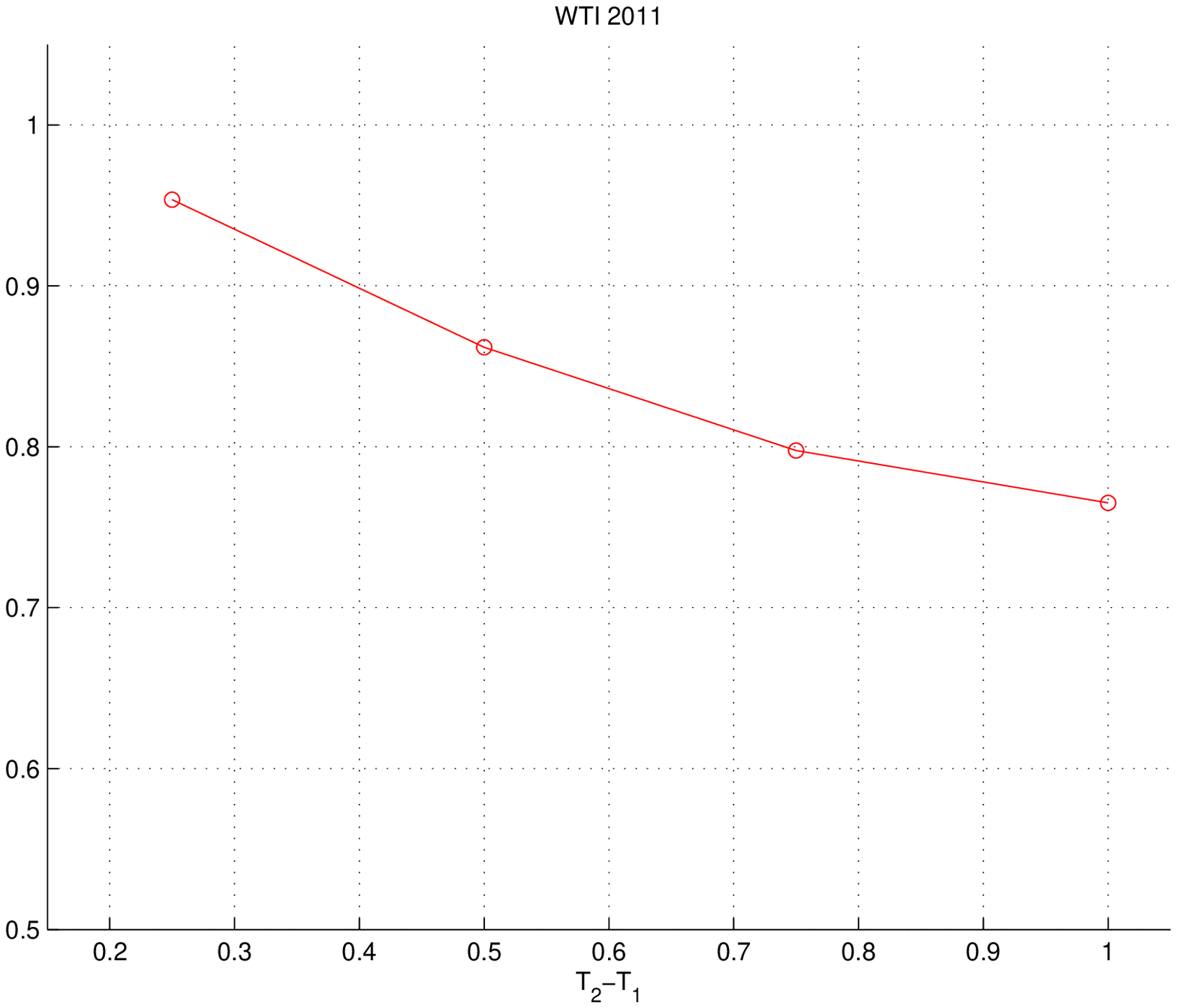}

	\caption{\label{Fig:clibratedimpliedcorrel}
Implied correlations from spread option prices obtained with the SV2F model calibrated to market data of December 2008 and March 2011.
The considered spread options have a maturity $T$ and first underlying futures expiry $T_1$ fixed at $3$ months.
$T_2 - T_1$, the difference between the underlying futures expiries, ranges from $3$ months to $1$ year.
\textit{Left column}: implied correlation smiles.
\textit{Right column}: at-the-money implied correlation term-structure.}
\end{figure}




\begin{figure}[ht]
\centering
    \includegraphics[height=6.0cm]{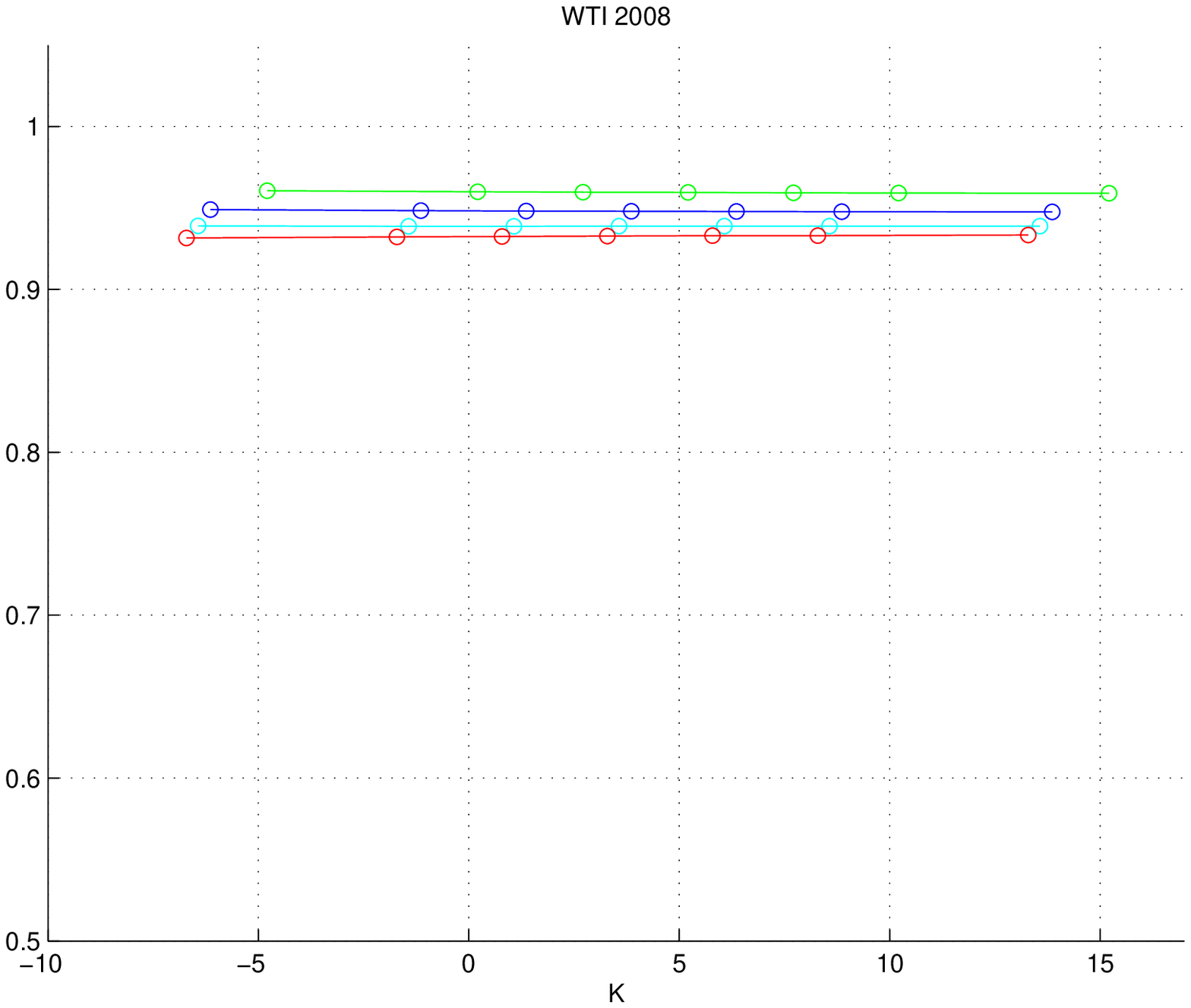} \quad
    \includegraphics[height=6.0cm]{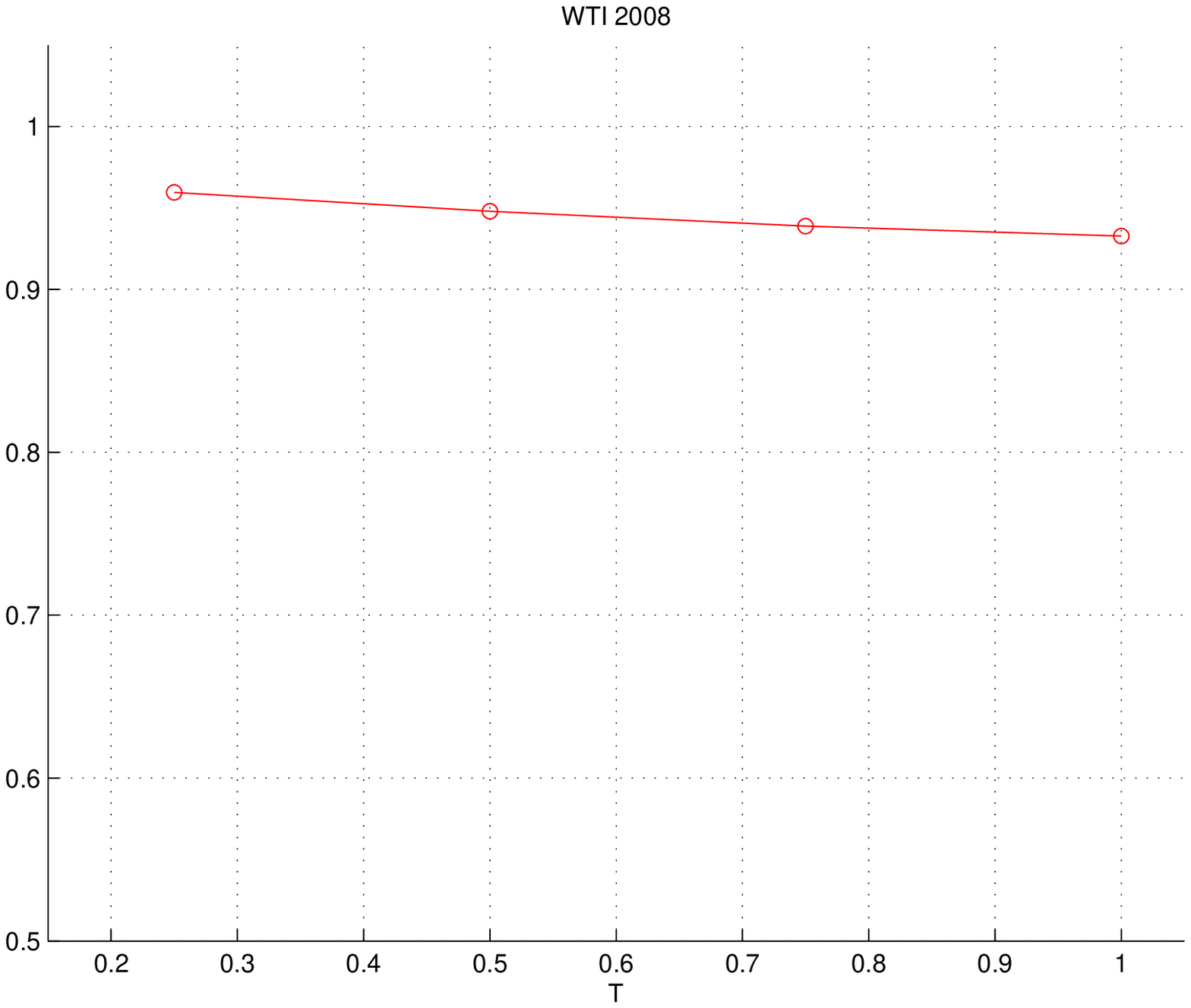}

    \includegraphics[height=6.0cm]{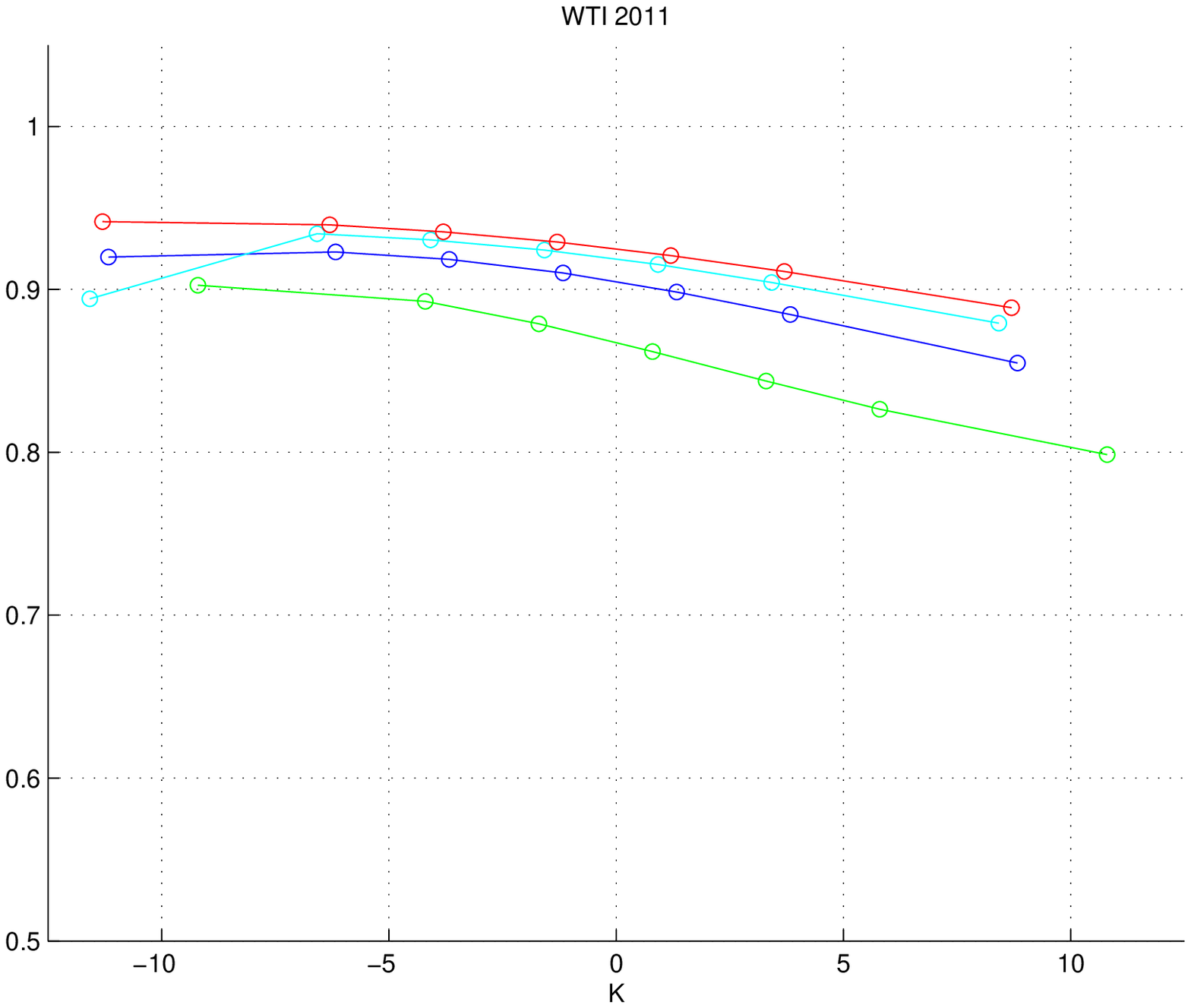} \quad
    \includegraphics[height=6.0cm]{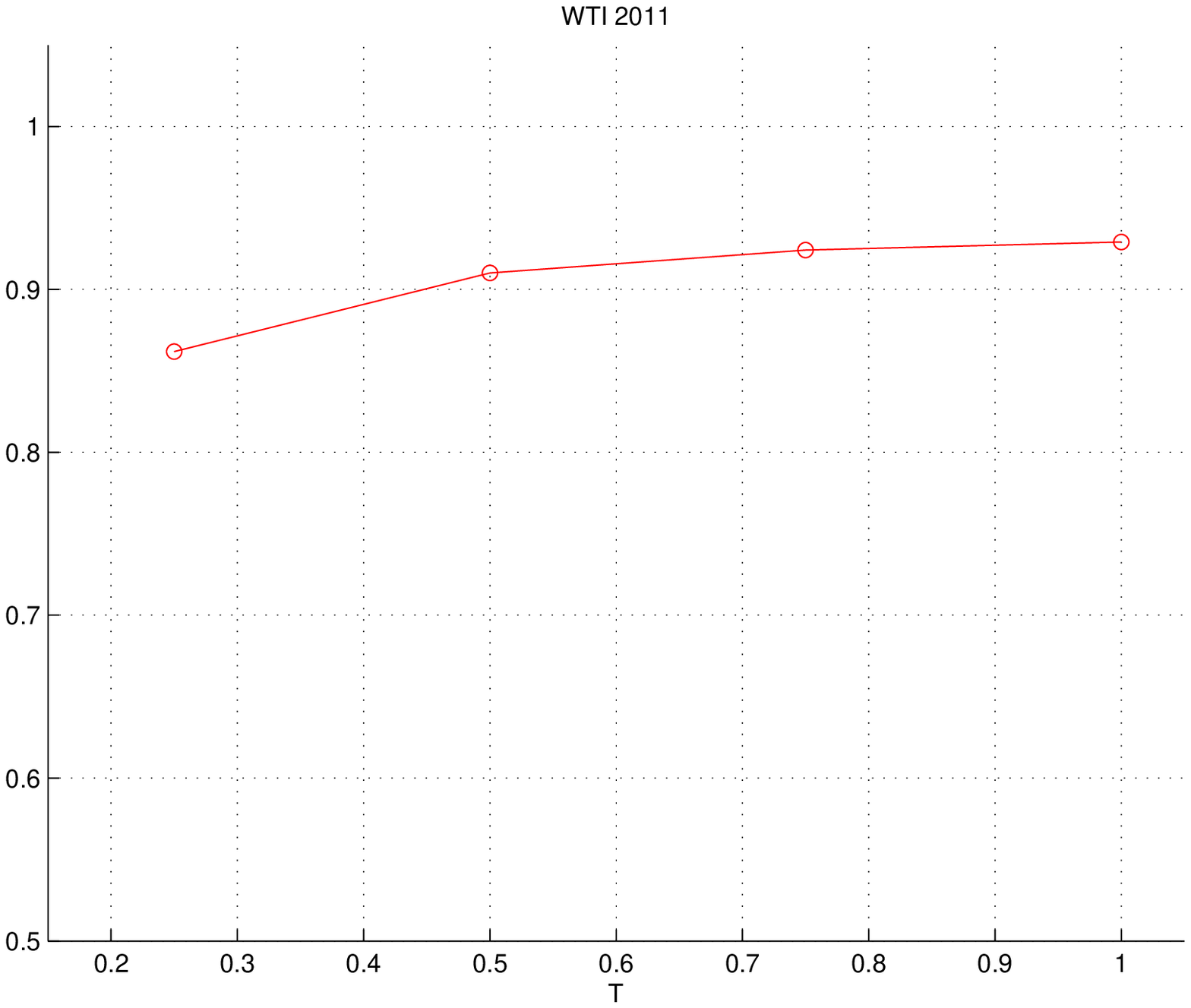}

    \caption{\label{Fig:clibratedimpliedcorrel2}
Implied correlations from spread option prices obtained with the SV2F model calibrated to market data of December 2008 and March 2011.
The considered spread options have a maturity $T$ and first underlying futures expiry $T_1$ varying from $3$ months to $1$ year.
$T_2 - T_1$, the difference between the underlying futures expiries, is held fixed at $6$ months.
\textit{Left column}: implied correlation smiles.
\textit{Right column}: at-the-money implied correlation term-structure.}
\end{figure}


\section{Conclusion}
\label{s:Conclusion}

We propose a multi-factor stochastic volatility model for commodity futures contracts.
In order to capture the Samuelson effect displayed by commodity futures contracts, we add expiry-dependent exponential damping factors to their volatility coefficients.
The pricing of single underlying European options on futures contracts is straightforward and can incorporate the volatility smile or skew observed in the market.
We calculate the joint characteristic function of two futures contracts in the model
and use the one-dimensional Fourier inversion method of \citet{CaldanaFusai2013} to price calendar spread options.
The model leads to stochastic correlation between the returns of two futures contracts.
We illustrate the distribution of this correlation in an example and compare it to the deterministic correlation of the corresponding \cite{ClewlowStrickland1999_2} model.
We also analyze the term-structure of dependence between pairs of futures in the model.
We do this by using suitable expressions to obtain the copula and copula density directly from the joint characteristic function.
When calibrated to vanilla options, the model is found to be able to produce stylized facts such as Samuelson effect and implied volatility smile
as well as a decreasing term-structure of dependence and implied correlation smile for spread options.

\pagebreak

\appendix


\section{Proofs}
\label{a1:Proofs}

In this appendix we prove Propositions \ref{Prop:JointCharacteristicFunction} and \ref{Prop:JointCharacteristicFunction_ClewlowStrickland}
by showing how to obtain the joint characteristic function $\phi$ of the log-returns $X_1(T)$ and $X_2(T)$ in the stochastic volatility model and the Clewlow-Strickland model.
We also prove Lemma \ref{Lemma:JointCDFFormula} and Proposition \ref{Prop:CopulaFromCharFunc}, which shows how the copula and its density can be calculated
from the joint characteristic function.

\begin{MyProof}{of Proposition \ref{Prop:JointCharacteristicFunction}.}
We have
\begin{align*}
\phi(u) &= \phi(u; T, T_1, T_2)
\\
&= {\mathbb E} \left[ \exp \left( i \sum_{k=1}^2 u_k X_k(T) \right) \right]
\\
&=
{\mathbb E}
\left[ \exp
\left( i \sum_{k=1}^2 u_k
\left\{
\sum_{j=1}^n \int_0^T e^{-\lambda_j(T_k - t)} \sqrt{v_j(t)} dB_j(t)
-\frac{1}{2} \sum_{j=1}^n \int_0^T e^{-2\lambda_j(T_k - t)} v_j(t) dt
\right\}
\right)
\right]
\\
&=
\prod_{j=1}^n
{\mathbb E}
\left[ \exp
\left( i \sum_{k=1}^2 u_k
\left\{
\int_0^T e^{-\lambda_j(T_k - t)} \sqrt{v_j(t)} dB_j(t)
-\frac{1}{2} \int_0^T e^{-2\lambda_j(T_k - t)} v_j(t) dt
\right\}
\right)
\right]
\\
&=
\prod_{j=1}^n E_j(u,T),
\end{align*}
where $E_j$ is a function of $u$ and $T$ given by
\begin{equation*}
E_j(u,T) = {\mathbb E}
\left[ \exp
\left( i \sum_{k=1}^2 u_k
\left\{
\int_0^T e^{-\lambda_j(T_k - t)} \sqrt{v_j(t)} dB_j(t)
-\frac{1}{2} \int_0^T e^{-2\lambda_j(T_k - t)} v_j(t) dt
\right\}
\right)
\right]
\end{equation*}
that otherwise depends only on the $j$-th model parameters $\lambda_j, \kappa_j, \theta_j, \sigma_j, v_{j,0}, \rho_j$.

We now calculate the function $E_j$. Since we are considering a fixed value of $j$, we drop this subscript in the following calculations.
We also write $\tilde{B}$ for $B_{n+j}$, the Brownian motion driving the $j$-th variance process.
Then we can decompose $B = B_j$, whose correlation with $\tilde{B} = B_{n+j}$ is given in equation \eqref{FuturesVarianceCorrelations} by
$\langle B_j, B_{n+j} \rangle = \rho_j dt = \rho dt$,
as $B = \rho \tilde{B} + \sqrt{1 - \rho^2} \hat{B}$, where $\hat{B}$ is uncorrelated with $\tilde{B}$.
Define the functions $f_1, f_2$ and $q$ given by
\begin{align*}
f_1(u,t) &= \sum_{k=1}^2 u_k e^{- \lambda(T_k - t)}, \quad f_2(u,t) = \sum_{k=1}^2 u_k e^{-2\lambda(T_k - t)},
\\
q(u,t)   &= i \rho \frac{\kappa - \lambda}{\sigma} f_1(u,t) - \frac{1}{2} (1 - \rho^2) f_1^2(u,t) - \frac{1}{2} i f_2(u,t).
\end{align*}
For simplicity, we write $f_1(t)$ for $f_1(u,t), f_2(t)$ for $f_2(u,t)$ and $q(t)$ for $q(u,t)$ in the following.

We first need an auxiliary result in order to calculate the characteristic function.
\begin{lemma}
\label{Lemma:StochasticIntegral}
\begin{equation}
\label{StochasticIntegral}
\sigma \int_0^T f_1(t) \sqrt{v(t)} d\tilde{B}(t)
=
\left[ f_1(t) \left\{ v(t) - \frac{\kappa \theta}{\lambda} \right\} \right]_0^T
+ (\kappa - \lambda) \int_0^T f_1(t) v(t) dt.
\end{equation}
\end{lemma}
\begin{proof}
Multiplying equation \eqref{VarianceSDE} by $f_1(t)$ and then integrating from $0$ to $T$ gives
\begin{equation}
\label{dv-integral1}
\int_0^T f_1(t) dv(t) = \int_0^T f_1(t) \kappa (\theta - v(t)) dt + \sigma \int_0^T f_1(t) \sqrt{v(t)} d\tilde{B}(t).
\end{equation}
Using It\^{o}-integration by parts (see \cite{Oksendal2003}), we also have
\begin{align}
\int_0^T f_1(t) dv(t)
&= \left[ f_1(t) v(t) \right]_0^T - \int_0^T v(t) df_1(t)
\nonumber
\\
&= \left[ f_1(t) v(t) \right]_0^T - \lambda \int_0^T f_1(t) v(t) dt.
\label{dv-integral2}
\end{align}
Equating the right hand sides of equations \eqref{dv-integral1} and \eqref{dv-integral2} gives
\begin{align*}
\sigma \int_0^T f_1(t) \sqrt{v(t)} d\tilde{B}(t)
&= \left[ f_1(t) v(t) \right]_0^T - \lambda \int_0^T f_1(t) v(t) dt - \int_0^T f_1(t) \kappa (\theta - v(t)) dt
\\
&= \left[ f_1(t) v(t) \right]_0^T - \kappa \theta \int_0^T f_1(t) dt + (\kappa - \lambda) \int_0^T f_1(t) v(t) dt
\\
&= \left[ f_1(t) \left\{ v(t) - \frac{\kappa \theta}{\lambda} \right\} \right]_0^T + (\kappa - \lambda) \int_0^T f_1(t) v(t) dt,
\end{align*}
which proves the lemma.
\end{proof}

We now calculate $E(u,T)$.
\begin{align*}
E(u,T)
&=
\left[ \exp
\left( i \sum_{k=1}^2 u_k
\left\{
\int_0^T e^{-\lambda(T_k - t)} \sqrt{v(t)} dB(t)
-\frac{1}{2} \int_0^T e^{-2\lambda(T_k - t)} v(t) dt
\right\}
\right)
\right]
\\
&=
{\mathbb E}
\left[
\exp \left( i \int_0^T f_1(t) \sqrt{v(t)} dB(t) -\frac{1}{2} i \int_0^T f_2(t) v(t) dt \right)
\right]
\\
&=
{\mathbb E}
\left[
\exp \left( i \rho \int_0^T f_1(t) \sqrt{v(t)} d\tilde{B}(t)
+ i \sqrt{1 - \rho^2} \int_0^T f_1(t) \sqrt{v(t)} d\hat{B}(t)
-\frac{1}{2} i \int_0^T f_2(t) v(t) dt \right)
\right]
\\
&=
{\mathbb E}
\left[
\exp \left( i \rho \int_0^T f_1(t) \sqrt{v(t)} d\tilde{B}(t)
-\frac{1}{2} (1 - \rho^2) \int_0^T \left( f_1(t) \right)^2 v(t) dt
-\frac{1}{2} i \int_0^T f_2(t) v(t) dt \right)
\right]
\\
&=
{\mathbb E}
\Big[
\exp \Big(
i \frac{\rho}{\sigma} \left[ f_1(t) \left\{ v(t) - \frac{\kappa \theta}{\lambda} \right\} \right]_0^T
+ i \rho \frac{\kappa - \lambda}{\sigma} \int_0^T f_1(t) v(t) dt
\\
&\quad -\frac{1}{2} (1 - \rho^2) \int_0^T \left( f_1(t) \right)^2 v(t) dt
-\frac{1}{2} i \int_0^T f_2(t) v(t) dt \Big)
\Big]
\\
&=
\exp \left( i \frac{\rho}{\sigma} \left\{ \frac{\kappa \theta}{\lambda} (f_1(0) - f_1(T)) - f_1(0) v(0) \right\} \right)
\\
&\quad \cdot
{\mathbb E}
\left[
\exp \left(
i \frac{\rho}{\sigma} f_1(T) v(T) + \int_0^T q(t) v(t) dt
\right)
\right].
\end{align*}
The expectation in the last line can be computed using the Feynman-Kac theorem (see \citet{Oksendal2003}).
Define the function $h$ given by
\begin{equation*}
h(t,v) = {\mathbb E} \left[ \exp \left( i \frac{\rho}{\sigma} f_1(T) v(T) + \int_t^T q(s) v(s) ds \right) \right].
\end{equation*}
Then $h$ satisfies the PDE
\begin{equation}
\label{hPDE}
\frac{\partial h}{\partial t} (t,v)
+ \kappa (\theta - v(t)) \frac{\partial h}{\partial v} (t,v)
+ \frac{1}{2} \sigma^2 v(t) \frac{\partial^2 h}{\partial v^2} (t,v)
+ q(t) v(t) h(t,v)
= 0,
\end{equation}
with terminal condition
\begin{equation*}
 h(T,v) = \exp \left( i \frac{\rho}{\sigma} f_1(T) v(T) \right).
\end{equation*}

We know from \citet{DuffiePanSingleton2000} that $h$ has affine form
\begin{equation}
\label{hGuess}
h(t,v) = \exp \left( A(t,T) v(t) + B(t,T) \right),
\end{equation}
with $A(T,T) = i \frac{\rho}{\sigma} f_1(T), B(T,T) = 0.$
Putting \eqref{hGuess} in \eqref{hPDE} gives
\begin{equation*}
B_t + A_t v + \kappa (\theta - v) A + \frac{1}{2} \sigma^2 v A^2 + q v = 0,
\end{equation*}
and collecting the terms with and without $v$ leads to the two ODEs
\begin{align}
\label{A_RiccatiEquation}
A_t - \kappa A + \frac{1}{2} \sigma^2 A^2 + q &= 0,
\\
\label{A_Primitive}
B_t + \kappa \theta A &= 0.
\end{align}
This completes the proof of the proposition.
\end{MyProof}


\begin{MyProof}{of Proposition \ref{Prop:JointCharacteristicFunction_ClewlowStrickland}.}
We calculate the joint characteristic function in the \citet{ClewlowStrickland1999_2} model as follows.
\begin{align*}
\phi(u) &= \phi(u; T, T_1, T_2)
\\
&= {\mathbb E} \left[ \exp \left( i \sum_{k=1}^2 u_k X_k(T) \right) \right]
\\
&=
{\mathbb E}
\left[ \exp
\left( i \sum_{k=1}^2 u_k
\left\{
\sum_{j=1}^n \int_0^T e^{-\lambda_j(T_k - t)} \sigma_j dB_j(t)
-\frac{1}{2} \sum_{j=1}^n \int_0^T e^{-2\lambda_j(T_k - t)} \sigma_j^2 dt
\right\}
\right)
\right]
\\
&=
\prod_{j=1}^n
{\mathbb E}
\left[ \exp
\left( i \sum_{k=1}^2 u_k
\left\{
\int_0^T e^{-\lambda_j(T_k - t)} \sigma_j dB_j(t)
-\frac{1}{2} \int_0^T e^{-2\lambda_j(T_k - t)} \sigma_j^2 dt
\right\}
\right)
\right]
\\
&=
\prod_{j=1}^n
\exp \left( i \sum_{k=1}^2 u_k
\left\{ -\frac{1}{2} \int_0^T e^{-2\lambda_j(T_k - t)} \sigma_j^2 dt \right\}
\right)
{\mathbb E}
\left[ \exp
\left( i \sum_{k=1}^2 u_k
\left\{ \int_0^T e^{-\lambda_j(T_k - t)} \sigma_j dB_j(t) \right\}
\right)
\right]
\\
&=
\prod_{j=1}^n
\exp
\left( i \sum_{k=1}^2 u_k
\left[
-\frac{\sigma_j^2}{4 \lambda_j} e^{-2\lambda_j(T_k - t)}
\right]_0^T
\right)
\exp \left(
-\frac{\sigma_j^2}{4 \lambda_j}
\left[
\left(
\sum_{k=1}^2 u_k e^{-\lambda_j(T_k - t)}
\right)^2
\right]_0^T
\right)
\\
&=
\prod_{j=1}^n
\exp \left(
-\frac{\sigma_j^2}{4 \lambda_j} (e^{2 \lambda_j T} - 1)
\left\{
i (u_1 e^{-2 \lambda_j T_1} + u_2 e^{-2 \lambda_j T_2}) + (u_1 e^{-\lambda_j T_1} + u_2 e^{-\lambda_j T_2})^2
\right\}
\right).
\end{align*}
This completes the proof of the proposition.
\end{MyProof}


\begin{MyProof}{of Lemma \ref{Lemma:JointCDFFormula}.}

The proof to obtain this expression is the same, \textit{mutatis mutandis}, as the proof in the univariate case provided in \citet{LeCourtoisWalter2014}.
Within the proof and for ease of reading, we drop the explicit dependencies on $T,T_1$ and $T_2$.

Let $a_1>0$ and $a_2 > 0$ be fixed and $h$ be the function defined by
\begin{equation*}
h(x_1,x_2) = e^{-(a_1x_1+a_2x_2)}G(x_1,x_2) = e^{-(a_1x_1+a_2x_2)}\int^{x_2}_{-\infty}\int^{x_1}_{-\infty}{g(s_1,s_2)ds_1ds_2}.
\end{equation*}
Now let $\Lambda$ be the two-dimensional Fourier Transform of $h$. We have
\begin{align*}
\Lambda(u_1,u_2) & = \int^{+\infty}_{-\infty}\int^{+\infty}_{-\infty}{e^{i(u_1 x_{1}+u_2 x_{2})}h(x_1,x_2)dx_1 dx_2}, \\
 & = \int^{+\infty}_{-\infty}\int^{+\infty}_{-\infty}{e^{i(u_1 x_{1}+u_2 x_{2})}\left(e^{-(a_1x_1+a_2x_2)}\int^{x_2}_{-\infty}\int^{x_1}_{-\infty}{g(s_1,s_2)ds_1ds_2} \right) dx_1 dx_2}, \\
 & = \int^{+\infty}_{-\infty}\int^{+\infty}_{-\infty}{ \int^{x_2}_{-\infty}\int^{x_1}_{-\infty}{e^{i(u_1 x_{1}+u_2 x_{2})} e^{-(a_1x_1+a_2x_2)}g(s_1,s_2)ds_1ds_2} dx_1 dx_2}.
\end{align*}
Noting that $-\infty<s_1<x_1<+\infty$ and $-\infty<s_2<x_2<+\infty$, the expression of $\Lambda$ becomes
\begin{align*}
\Lambda(u_1,u_2) & = \int^{+\infty}_{-\infty}\int^{+\infty}_{-\infty}{ \int^{+\infty}_{s_2}\int^{+\infty}_{s_1}{e^{i(u_1 x_{1}+u_2 x_{2})} e^{-(a_1x_1+a_2x_2)}g(s_1,s_2)dx_1dx_2} ds_1 ds_2}, \\
 & = \int^{+\infty}_{-\infty}\int^{+\infty}_{-\infty}{ g(s_1,s_2) \left( \int^{+\infty}_{s_2}\int^{+\infty}_{s_1}{e^{i(u_1 x_{1}+u_2 x_{2})} e^{-(a_1x_1+a_2x_2)}dx_1dx_2} \right) ds_1 ds_2}.
\end{align*}
The double integral between parentheses can be computed as
\begin{align*}
\int^{+\infty}_{s_2}\int^{+\infty}_{s_1}{e^{i(u_1 x_{1}+u_2 x_{2})} e^{-(a_1x_1+a_2x_2)}dx_1dx_2} & = \int^{+\infty}_{s_1}{e^{iu_1 x_{1}} e^{-a_1x_1} dx_1} \int^{+\infty}_{s_2}{e^{iu_2 x_{2}} e^{-a_2x_2} dx_2}, \\
& = \left[ \frac{e^{-(a_1-iu_1)x_1}}{-(a_1-iu_1)} \right]^{+\infty}_{s_1} \left[ \frac{e^{-(a_2-iu_2)x_2}}{-(a_2-iu_2)} \right]^{+\infty}_{s_2}.
\end{align*}
Note that $\left|e^{-(a_1-iu_1)x_1} \right| \longrightarrow 0$ when $x_1$ goes to $+\infty$ and $\left|e^{-(a_2-iu_2)x_2} \right| \longrightarrow 0$
when $x_2$ goes to $+\infty$, so that we obtain
\begin{align*}
\Lambda(u_1,u_2) & = \int^{+\infty}_{-\infty}\int^{+\infty}_{-\infty}{ g(s_1,s_2) \left(-\frac{e^{-(a_1-iu_1)s_1}}{-(a_1-iu_1)}\right) \left(-\frac{e^{-(a_2-iu_2)s_2}}{-(a_2-iu_2)}\right) ds_1 ds_2}, \\
 & = \frac{1}{(a_1-iu_1)(a_2-iu_2)}\int^{+\infty}_{-\infty}\int^{+\infty}_{-\infty}{ g(s_1,s_2) e^{-(a_1-iu_1)s_1}e^{-(a_2-iu_2)s_2} ds_1 ds_2}, \\
 & = \frac{1}{(a_1-iu_1)(a_2-iu_2)}\int^{+\infty}_{-\infty}\int^{+\infty}_{-\infty}{ g(s_1,s_2) e^{i \left((u_1+ia_1)s_1+(u_2+ia_2)s_2 \right)} ds_1 ds_2}, \\
 & = \frac{\phi(u_1+ia_1,u_2+ia_2)}{(a_1-iu_1)(a_2-iu_2)},
\end{align*}
since
\begin{equation*}
\phi(u_1,u_2) = \int^{+\infty}_{-\infty}\int^{+\infty}_{-\infty}{ e^{i(u_1s_1+u_2s_2)} g(s_1,s_2) ds_1 ds_2}.
\end{equation*}
The function $h$ can be written as the two-dimensional inverse Fourier Transform of $\Lambda$:
\begin{equation*}
h(x_1,x_2) = \frac{1}{4 \pi^2} \int^{+\infty}_{-\infty}\int^{+\infty}_{-\infty}{e^{-i(u_1x_1+u_2x_2)}\frac{\phi(u_1+ia_1,u_2+ia_2)}{(a_1-iu_1)(a_2-iu_2)} du_1 du_2}, \\
\end{equation*}
and $G$ is then easily obtained as
\begin{equation*}
G(x_1,x_2) = \frac{e^{a_1x_1+a_2x_2}}{4 \pi^2}
\int^{+\infty}_{-\infty}\int^{+\infty}_{-\infty}{e^{-i(u_1x_1+u_2x_2)}\frac{\phi(u_1+ia_1,u_2+ia_2)}{(a_1-iu_1)(a_2-iu_2)} du_1 du_2},
\end{equation*}
which concludes the proof.
\end{MyProof}


\begin{MyProof}{of Proposition \ref{Prop:CopulaFromCharFunc}.}

Sklar's Theorem allows one to write the copula function of a pair of random variables from its joint distribution function as, for $(v_1,v_2) \in [0,1]^2$,
\begin{equation*}
C(v_1,v_2,T)=G\left(G^{-1}_1(v_1,T),G^{-1}_2(v_2,T),T \right).
\end{equation*}

The expression for the copula function in Proposition \ref{Prop:CopulaFromCharFunc} follows by using Lemma \ref{Lemma:JointCDFFormula},
which expresses the joint distribution function in terms of the joint characteristic function $\phi$.
Assuming $C(.,T)$ is absolutely continuous, we can write its copula density, for $(v_1,v_2) \in [0,1]^2$, as
\begin{equation*}
c(v_1,v_2,T)=\frac{g(G^{-1}_1(v_1,T),G^{-1}_2(v_2,T),T)}{g_1(G^{-1}_1(v_1,T),T)g_2(G^{-1}_2(v_2,T),T)}.
\end{equation*}

Again, the expression for the copula density in Proposition \ref{Prop:CopulaFromCharFunc} follows by using expressions (\ref{JointPDFFormula}),
(\ref{MarginalPDF1Formula}) and (\ref{MarginalPDF2Formula}) that express the joint and marginal densities of $(X_1(T),X_2(T))$ in terms of the joint characteristic function $\phi$.

\end{MyProof}


\bibliographystyle{plainnat}

\bibliography{articles,books,websites}

\end{document}